\documentclass[11pt,a4paper]{article}

\usepackage[T1]{fontenc}
\usepackage[utf8]{inputenc}
\usepackage[english]{babel}
\usepackage[a4paper,margin=2.5cm]{geometry}
\usepackage{lmodern}
\usepackage{microtype}
\usepackage{setspace}
\usepackage{mathtools,amssymb,amsthm}
\usepackage{booktabs,array,tabularx,siunitx}
\usepackage{graphicx}
\graphicspath{{}}%figures/
\usepackage{caption}
\usepackage{enumitem}
\usepackage{placeins}
\usepackage{hyperref}
\usepackage[nameinlink,noabbrev]{cleveref}

\allowdisplaybreaks[2]
\setlist{leftmargin=2em,itemsep=0.2em,topsep=0.35em}
\hypersetup{colorlinks=true,linkcolor=black,citecolor=black,urlcolor=black,
	pdftitle={Adaptive singular-point methods for pricing and hedging surrenderable equity-linked contracts},
	pdfauthor={Andrea Molent; Marcellino Gaudenzi},
	pdfkeywords={dynamic programming, stochastic models, life insurance, early exercise, adaptive approximation}}

\newtheorem{theorem}{Theorem}[section]
\newtheorem{lemma}[theorem]{Lemma}
\newtheorem{proposition}[theorem]{Proposition}
\newtheorem{corollary}[theorem]{Corollary}
\newtheorem{remark}[theorem]{Remark}
\newcommand{\supp}{Supplementary Material}

\title{Adaptive singular-point method for pricing and hedging\\
	surrenderable equity-linked contracts}
\author{Andrea Molent\thanks{Dipartimento di Scienze Economiche e Statistiche,
		Universit\`a degli Studi di Udine, Udine, Italy. Email:
		\href{mailto:andrea.molent@uniud.it}{andrea.molent@uniud.it}. ORCID:
		\href{https://orcid.org/0000-0002-0887-826X}{0000-0002-0887-826X}.}
	\and Marcellino Gaudenzi\thanks{Dipartimento di Scienze Economiche e Statistiche,
		Universit\`a degli Studi di Udine, Udine, Italy. Email:
		\href{mailto:marcellino.gaudenzi@uniud.it}{marcellino.gaudenzi@uniud.it}. ORCID:
		\href{https://orcid.org/0000-0003-0366-3435}{0000-0003-0366-3435}.}}
\date{}

\begin{document}
	\onehalfspacing
	\maketitle
	
	\begin{abstract}
		We propose a deterministic numerical method for pricing and hedging
		surrenderable equity-linked life-insurance contracts with periodic premiums
		and fund contributions, maturity and death guarantees, and Bermudan surrender
		under correlated stochastic volatility and stochastic interest rates. The main
		computational challenge is the non-recombining accumulated fund, which couples
		with multiple stochastic factors and early exercise. Our key idea is to avoid
		a full multidimensional fund lattice: variance and interest-rate factors are
		discretized on recombining lattices, while at each factor node the contract
		value is represented as an adaptive one-dimensional function of the fund.
		Periodic contributions then act as translations of the fund argument, whereas
		surrender is handled directly through a backward obstacle condition.
		Piecewise-cubic representations propagate payoff and exercise singularities
		and are compressed by continuous pruning criteria that control the
		representation error. We establish weak convergence of the financial chains,
		convergence of the adaptive valuation under vanishing representation error,
		and Delta consistency on regular fund regions. For the strict binomial scheme,
		additional regularity yields first-order weak accuracy and a Talay--Tubaro
		expansion supporting Richardson extrapolation. Numerical experiments show
		compact representations, favorable cost--accuracy, and close agreement with
		independent Monte Carlo and cross-fitted least-squares Monte Carlo benchmarks.
		Hedging results further show that contracts with similar values can generate
		materially different exposures to equity, volatility, and interest-rate risk.
	\end{abstract}
	
	
	\section{Introduction}\label{sec:introduction}
	
	Equity-linked life-insurance contracts combine market participation with
	long-dated guarantees and, frequently, a right to surrender before maturity.
	Related unit-linked products remain economically significant in Europe:
	EIOPA reports about EUR~42 billion of gross written premium for these products
	in its 2024 sample \cite{EIOPA2026CPP}.
	These figures highlight the continuing practical
	relevance of reliable methods for the valuation, surrender analysis, and
	hedging of market-linked insurance contracts.
	
	The early literature established the valuation of asset-linked guarantees
	\cite{BrennanSchwartz1976} and later incorporated stochastic interest rates
	\cite{NielsenSandmann1995}, surrender and bonus mechanisms
	\cite{GrosenJorgensen2000,ShenXu2005}, and periodic-premium contracts
	\cite{Bacinello2005}. These features are economically important but create a
	challenging sequential decision problem: discounting, guarantee value, market
	volatility, mortality, future contributions, and the exercise decision evolve
	jointly over horizons of one or two decades.
	
	From an operational-research perspective, the valuation is a finite-horizon
	stochastic dynamic program with a continuous, non-recombining resource state.
	Recombining trees can represent the financial factors, but annual additions to
	the accumulated fund prevent the fund itself from recombining. A direct grid in
	all state variables is therefore expensive. Approximate dynamic programming
	usually addresses such state-space growth by replacing the exact value function
	with a structured approximation \cite{Powell2011}. Regression-based simulation
	is one important implementation, but in the present stopping problem it combines
	sampling, approximation, and stopping-rule errors
	\cite{BacinelloBiffisMillossovich2009,
		BacinelloBiffisMillossovich2010,LongstaffSchwartz2001}. Related work develops exact or error-controlled
	simulation schemes for stochastic-volatility models and path-dependent claims
	\cite{BrignoneJunike2026,LiWu2019}. Singular-point methods offer a deterministic,
	structure-exploiting alternative: they retain the value function only where
	contractual kinks or the backward recursion generate relevant nonlinear
	structure \cite{CostabileGaudenziMassaboZanette2009,CostabileMassaboRusso2008,
		GaudenziZanetteLepellere2010}.
	
	We develop this idea for a fully correlated Heston--CIR++ model. Heston
	stochastic variance captures leverage and time-varying equity risk
	\cite{Heston1993}; the CIR++ rate specification preserves a non-negative
	square-root factor while fitting the observed initial term structure
	\cite{BrigoMercurio2001,CoxIngersollRoss1985}. Hybrid equity--rate models show that volatility, rates, and dependence can
	materially affect long-dated prices and risk exposures, including in settings
	with optimal surrender \cite{GoudenegeMolentWeiZanette2025,
		GoudenegeMolentZanette2016,GoudenegeMolentZanette2026,
		GrzelakOosterlee2011,KangZiveyi2018}.
	
	Within the proposed scheme, the variance and CIR rate factors are discretized
	jointly on recombining lattices so as to preserve their dependence structure,
	while at each joint factor node the contract value is represented as an
	adaptive one-dimensional function of the accumulated fund. Consequently, the
	contract value depends on the past only through the current fund, variance,
	and interest-rate states; no additional variable is needed to record earlier
	contributions. This representation naturally accommodates periodic
	contributions and surrender decisions at prescribed dates while concentrating
	computational effort where the value function exhibits relevant nonlinear
	features.
	
	The contributions of the paper are threefold. First, we represent the non-recombining fund dimension through an adaptive value function rather than a fixed global grid. Periodic contributions then act simply by shifting the fund argument, while annual surrender decisions can be incorporated directly within the same deterministic backward recursion. This preserves the relevant nonlinear structure of the contract value while avoiding the state-space growth associated with a full discretization of the fund variable.
	
	Second, we extend singular-point dynamic programming to a fully correlated Heston--CIR++ setting by combining recombining discretizations of stochastic variance and interest rates with the adaptive representation of the fund dimension. The resulting construction retains a deterministic backward recursion while accommodating stochastic volatility, stochastic rates, their dependence, mortality, and Bermudan surrender.
	
	Third, we develop a convergence and error-control framework that distinguishes the financial time-discretization error from the error introduced by the adaptive representation and the numerical boundaries. We establish convergence of the contract values and, under explicit regularity and representation conditions, consistency of the fund Delta. For the binomial discretization, stronger regularity assumptions also yield a first-order weak-error expansion and a theoretically supported Richardson benchmark.
	
	The numerical study assesses accuracy and computational efficiency against independent Monte Carlo and least-squares Monte Carlo benchmarks, and uses the resulting framework to investigate surrender behavior, model effects, and multifactor hedging.
	
	The remainder of the paper is organized as follows. Section~\ref{sec:model}
	defines the model and contract. Section~\ref{sec:method} presents the numerical
	method, Section~\ref{sec:convergence} states the convergence and error-control
	results, and Section~\ref{sec:numerics} reports the numerical study.
	
	\section{Model and contract}\label{sec:model}
	
	We consider an equity-linked life-insurance contract in which the policyholder
	pays premiums over time and deterministic amounts are credited to an investment
	fund linked to an equity index. The contract provides guarantees at death and
	maturity and may also allow the policyholder to surrender at predetermined
	dates. If the contract is continued at such a date, the scheduled premium and
	fund contribution are paid and the fund subsequently evolves with the financial
	market until the next contractual event.
	
	For the specification considered below, the contract has maturity
	$T\in\mathbb N$ years, with policy anniversaries indexed by
	$i=0,\ldots,T$, and $a_0$ denotes the issue age. We assume, for simplicity,
	that premiums and fund contributions are paid annually at policy anniversaries
	and remain constant over time: $P$ denotes the annual premium and $D$ the
	corresponding amount credited to the fund. Other deterministic payment
	frequencies can be accommodated by including the corresponding dates in the
	time grid.
	
	\subsection{Financial dynamics and curve fit}
	Under the risk-neutral measure $\mathbb Q$, the equity index $S$, instantaneous variance factor $V$,
	and non-negative rate factor $X$ satisfy
	\begin{align}
		{\mathrm dS_t}&=(r_t-q)\,S_t\,\mathrm dt+\sqrt{V_t}\,S_t\,\mathrm dW_t^S,\\
		\mathrm dV_t&=\kappa_V(\theta_V-V_t)\,\mathrm dt+\sigma_V\sqrt{V_t}\,\mathrm dW_t^V,\\
		\mathrm dX_t&=\kappa_r(\theta_r-X_t)\,\mathrm dt+\sigma_r\sqrt{X_t}\,\mathrm dW_t^r,\\
		r_t&=X_t+\varphi(t).
		\label{eq:model-dynamics}
	\end{align}
	The Brownian correlation matrix is
	\begin{equation}
		\mathbf R=\begin{pmatrix}
			1&\rho_{SV}&\rho_{Sr}\\
			\rho_{SV}&1&\rho_{Vr}\\
			\rho_{Sr}&\rho_{Vr}&1
		\end{pmatrix},
		\qquad
		1+2\rho_{SV}\rho_{Sr}\rho_{Vr}-\rho_{SV}^2-\rho_{Sr}^2-\rho_{Vr}^2\ge0.
		\label{eq:correlation-matrix}
	\end{equation}
	Between contract dates, the accumulated fund $F$ has the same proportional return as
	$S$,
	\begin{equation}
		{\mathrm dF_t}=(r_t-q)\,{F_t}\,\mathrm dt+\sqrt{V_t}\,{F_t}\,\mathrm dW_t^S,
		\qquad F_{i}^{+}=F_{i}^{-}+D,\quad i=0,\ldots,T-1.
		\label{eq:fund-dynamics}
	\end{equation}
	Here $F_{i}^{-}$ and $F_{i}^{+}$ denote the fund values immediately before and
	after the contribution at anniversary $i$, respectively.
	The jump relation applies whenever the contract is in force and, at a
	surrender anniversary, is continued.
	The deterministic shift $\varphi$ in \eqref{eq:model-dynamics} is chosen so that the model reproduces the
	initial market discount curve exactly. Let $P^{\mathrm{CIR}}(0,u)$ denote the
	zero-coupon bond price generated by the unshifted CIR factor $X$. Since
	$r_t=X_t+\varphi(t)$, the model zero-coupon price with maturity $u$ is
	\[
	P(0,u)
	=
	P^{\mathrm{CIR}}(0,u)
	\exp\!\left(-\int_0^u \varphi(s)\,\mathrm ds\right).
	\]
	Therefore, fitting the initial market discount curve exactly requires
	\begin{equation}
		P^M(0,u)
		=
		P^{\mathrm{CIR}}(0,u)
		\exp\!\left(-\int_0^u\varphi(s)\,\mathrm ds\right),
		\qquad 0\le u\le T.
		\label{eq:cirpp-shift}
	\end{equation}
	The empirical curve is represented by a Svensson specification; its formula and
	parameters are reported in the \supp.
	
	\subsection{Annual premiums, guarantees, mortality, and surrender}
	
	At inception and, subsequently, at each anniversary
	$i=1,\ldots,T-1$ at which the contract is continued, the policyholder pays
	the annual premium $P$ and the corresponding amount $D$ is credited to the
	investment fund; no premium or contribution is due at maturity. The contract
	may terminate upon death, through surrender at a prescribed anniversary, or
	at maturity. No surrender decision is available at inception.
	
	Let $g_m$ denote the accumulation rate used for the maturity and death
	guarantees, and let $g_s$ denote the rate used for the guarantee component of
	the surrender payoff. For $\gamma\in\{g_m,g_s\}$ and $i=1,\ldots,T$, define
	\[
	G_i^{(\gamma)}
	=
	\sum_{j=0}^{i-1}D\exp\{\gamma(i-j)\}.
	\]
	For a contract still in force at anniversary $i$,
	$G_i^{(\gamma)}$ is the accumulated value of the contributions credited at
	anniversaries $0,\ldots,i-1$. If the policyholder is alive and the contract
	remains in force at maturity, the maturity benefit, as a function of the fund
	value immediately before settlement, is
	$\max\{F,G_T^{(g_m)}\}$.
	
	Mortality is assumed to be independent of the financial factors. Let $q_x$
	denote the one-year probability of death between ages $x$ and $x+1$,
	conditional on survival to age $x$. Hence, for $i=0,\ldots,T-1$,
	$q_{a_0+i}$ is the probability of death during policy year $(i,i+1]$,
	conditional on being alive at anniversary $i$. If the contract is in force
	immediately after anniversary $i$ and death occurs during that policy year,
	the death benefit paid at anniversary $i+1$ is
	$B_{i+1}(F)=\max\{F,G_{i+1}^{(g_m)}\}$, where $F$ denotes the fund value
	immediately before settlement. Since mortality is independent of the financial
	factors, the mortality probabilities enter the backward recursion as
	deterministic actuarial weights.
	
	At each surrender anniversary $i=1,\ldots,T-1$, conditional on the
	policyholder being alive and the contract being in force, the surrender
	decision is made before the new premium and fund contribution. Let $F_i^-$
	denote the fund value immediately before this decision and let
	$\alpha_i^s\in(0,1]$ be the fraction of the fund returned under the
	fund-based component. We consider the fund-based, guarantee-based, and mixed
	surrender payoffs
	\begin{equation} 
			R_{\mathrm{FB}}(i,F)  = \alpha_i^sF,\quad
			R_{\mathrm{GB}}(i,F)  = G_i^{(g_s)},\quad
			R_{\mathrm{MX}}(i,F)  = \max\{\alpha_i^sF,G_i^{(g_s)}\}. 
		\label{eq:surrender-payoffs}
	\end{equation}
	
	If surrender is exercised under contract $J$, the amount $R_J(i,F_i^-)$ is
	paid at anniversary $i$. No premium is paid and no fund contribution is
	credited at that anniversary, and the contract terminates. Because the
	decision precedes the contribution due at anniversary $i$,
	$G_i^{(g_s)}$ contains only the contributions previously credited at
	anniversaries $0,\ldots,i-1$.
	
	If the contract is continued, the policyholder pays the premium $P$ and the
	amount $D$ is credited to the fund, so that $F_i^+=F_i^-+D$. The financial
	state then evolves from $(F_i^+,V_i,X_i)$ over the following policy year
	$(i,i+1]$. If death occurs during that year,
	$B_{i+1}(F_{i+1}^-)$ is paid at anniversary $i+1$. If the policyholder
	survives, the contract proceeds to the next surrender anniversary, or to the
	maturity benefit when $i+1=T$.
	
	For a fixed annual premium $P$, let $U_i(F,v,x)$ denote the contract value
	immediately before the surrender decision at anniversary $i$, and let
	$C_i(F,v,x)$ denote the value obtained by choosing continuation from the same
	pre-decision state. Thus, for $i=1,\ldots,T-1$,
	\begin{equation}
		U_i(F,v,x)
		=
		\max\{C_i(F,v,x),R_J(i,F)\},
		\qquad
		J\in\{\mathrm{FB},\mathrm{GB},\mathrm{MX}\}.
		\label{eq:surrender-obstacle}
	\end{equation}
	
	Let $P_{\mathrm{NS}}$ denote the fair annual premium of the no-surrender
	contract, and let $P_J$, $J\in\{\mathrm{FB},\mathrm{GB},\mathrm{MX}\}$,
	denote the corresponding fair annual premiums of the surrenderable contracts.
	
	Conditional on the contract being in force, the valuation state, apart from
	time, is $(F,V,X)$. The current fund value incorporates all previously
	credited contributions and their subsequent investment performance, while the
	guarantee levels are deterministic functions of the anniversary index.
	Therefore, no additional state variable is required to record the contribution
	history.
	
	\section{Adaptive singular-point dynamic programming}\label{sec:method}
	Let $N_{\rm yr}$ be the number of numerical time steps per year, set
	$N=T N_{\rm yr}$ and $h=1/N_{\rm yr}=T/N$, and define
	$t_n=nh$, $n=0,\ldots,N$. Thus anniversary $i$ corresponds to the grid time
	$t_{iN_{\rm yr}}=i$. At time $t_n$, $(j,k)$ indexes a joint variance--rate
	node. For a generic square-root factor $Z\in\{V,X\}$,
	$(\kappa_Z,\theta_Z,\sigma_Z)$ denotes its mean-reversion speed, long-run mean,
	and volatility coefficient. For convenience, $D_n=D$ when $t_n$ is a payment
	anniversary $i=0,\ldots,T-1$, and $D_n=0$ otherwise; the premium outflow $P$
	likewise enters the continuation value only at payment anniversaries.
	
	\subsection{Recombining square-root lattices}
	
	Both CIR-type factors $Z\in\{V,X\}$ are represented on non-negative
	recombining square-root lattices. For parameters
	$(\kappa_Z,\theta_Z,\sigma_Z)$, the exact one-step conditional mean and
	variance are
	\begin{equation}
		m_Z(z)
		=
		\theta_Z+(z-\theta_Z)e^{-\kappa_Zh},
		\qquad
		v_Z(z)
		=
		\frac{\sigma_Z^2z}{\kappa_Z}e^{-\kappa_Zh}
		(1-e^{-\kappa_Zh})
		+
		\frac{\theta_Z\sigma_Z^2}{2\kappa_Z}
		(1-e^{-\kappa_Zh})^2.
		\label{eq:cir-exact-moments}
	\end{equation}
	The two lattice constructions use these moments differently. The
	adaptive-trinomial marginals match both $m_Z(z)$ and $v_Z(z)$ exactly.
	The binomial comparison matches $m_Z(z)$ exactly, while its conditional
	variance agrees with $v_Z(z)$ up to an $O(h^2)$ error on interior compact
	sets.
	
	\subsubsection{Binomial factor lattices}
	\label{subsubsec:binomial-factor-lattices}
	
	The fully binomial construction is retained as an independent deterministic
	comparison. For each factor, the grid spacing in the square-root coordinate
	is $	\Delta_Z=\frac{\sigma_Z}{2}\sqrt h,$
	and consecutive time layers use opposite grid parities. From a parent state
	$z$, let $z_d<z_u$ be the two neighboring nodes on the next time layer that
	bracket the exact conditional mean, $
	z_d\le m_Z(z)\le z_u$.
	The upper transition probability is determined by
	\begin{equation}
		p_{Z,h}(z)
		=
		\frac{m_Z(z)-z_d}{z_u-z_d},
		\qquad
		\mathbb{P}_h(Z_{n+1}=z_u\mid Z_n=z)=p_{Z,h}(z),
		\label{eq:binomial-mean-probability}
	\end{equation}
	with lower probability $1-p_{Z,h}(z)$. Consequently,
	$\mathbb{E}_h[Z_{n+1}\mid Z_n=z]=m_Z(z)$ exactly, while the bracketing
	condition guarantees non-negative marginal probabilities.
	
	At an interior state and for sufficiently small $h$, the selected successors
	are the adjacent opposite-parity nodes
	\begin{equation}
		z_\pm
		=
		\left(
		\sqrt z\pm\frac{\sigma_Z}{2}\sqrt h
		\right)^2.
		\label{eq:binomial-interior-successors}
	\end{equation}
	The corresponding discrete conditional variance is
	\begin{equation}
		\widehat v_{Z,h}(z)
		=
		p_{Z,h}(z)\bigl(1-p_{Z,h}(z)\bigr)(z_u-z_d)^2
		=
		v_Z(z)+O(h^2)
		\label{eq:binomial-marginal-variance}
	\end{equation}
	uniformly on compact subsets of $(0,\infty)$. Thus the binomial marginal
	matches the exact CIR mean but only the leading-order CIR variance. Near the
	lower boundary or on coarse grids, the bracketing rule, rather than the fixed
	adjacent pair in \eqref{eq:binomial-interior-successors}, determines the two
	successors.
	
	Let $
	p_V=p_{V,h}(v),
	\qquad
	p_X=p_{X,h}(x)$
	be the marginal upper probabilities at the current variance and rate-factor
	nodes. The joint $2\times2$ table is
	\begin{equation}
		\begin{aligned}
			p_{uu}&=p_Vp_X+\delta,
			&
			p_{ud}&=p_V(1-p_X)-\delta,
			\\
			p_{du}&=(1-p_V)p_X-\delta,
			&
			p_{dd}&=(1-p_V)(1-p_X)+\delta,
		\end{aligned}
		\label{eq:binomial-joint-table}
	\end{equation}
	where
	\begin{equation}
		\delta
		=
		\rho_{Vr}
		\sqrt{p_V(1-p_V)p_X(1-p_X)}.
		\label{eq:binomial-coupling}
	\end{equation}
	This construction preserves both marginal transition laws and imposes the
	correlation $\rho_{Vr}$ between their standardized binomial innovations.
	Equivalently, the discrete factor covariance is
	$\rho_{Vr}\sqrt{\widehat v_{V,h}(v)\widehat v_{X,h}(x)}$, and is therefore
	locally consistent with the target CIR covariance. 
	Non-negativity of all four joint probabilities is equivalent to the local
	Fr\'echet bounds
	\begin{equation}
		\max\{-p_Vp_X,-(1-p_V)(1-p_X)\}
		\le\delta\le
		\min\{p_V(1-p_X),(1-p_V)p_X\}.
		\label{eq:frechet-bounds-main}
	\end{equation}
	At a given time and parent state, let
	\[
	\mathcal{I}_h
	=
	\left\{
	\rho\in[-1,1]:
	p_{ab}(\rho)\ge0
	\text{ for all joint binomial branches }(a,b)
	\right\}
	\]
	denote the local admissible correlation interval. Whenever the requested
	$\rho_{Vr}$ lies outside $\mathcal{I}_h$, the binomial comparison replaces it
	by its projection onto this interval:
	\begin{equation}
		\rho_{Vr,h}^{\rm eff}
		=
		\operatorname{proj}_{\mathcal{I}_h}(\rho_{Vr})
		=
		\arg\min_{\rho\in\mathcal{I}_h}
		|\rho-\rho_{Vr}|.
		\label{eq:effective-rho}
	\end{equation}
	The dependence of $\mathcal{I}_h$ and $\rho_{Vr,h}^{\rm eff}$ on time and the
	parent state is suppressed in the notation. The projection changes only the
	dependence between the two factors and leaves both marginal transition laws
	unchanged. Its asymptotic irrelevance is analyzed in
	Section~\ref{subsec:projection-stability}.
	
	\subsubsection{Adaptive-trinomial factor lattices}
	\label{subsubsec:trinomial-factor-lattices}
	
	For the trinomial construction, the grid is uniform in the Lamperti
	coordinate $\chi_Z=2\sqrt Z/\sigma_Z$:
	\begin{equation}
		\chi_{Z,j}
		=
		\chi_{Z,0}+j\lambda_Z\sqrt h,
		\qquad
		Z_j
		=
		\left(
		\frac{\sigma_Z}{2}
		\max\{\chi_{Z,j},0\}
		\right)^2.
		\label{eq:trinomial-grid}
	\end{equation}
	The grid itself is fixed; adaptivity enters through the nodewise choice of
	the transition stencil. At each parent state $z$, the algorithm selects the
	narrowest ordered triplet of grid nodes, not necessarily consecutive, for
	which non-negative probabilities $p_i^Z$, $i=1,2,3$, satisfy
	\begin{equation}
		\sum_{i=1}^3p_i^Z=1,
		\qquad
		\sum_{i=1}^3p_i^ZZ_i'=m_Z(z),
		\qquad
		\sum_{i=1}^3p_i^Z
		\bigl(Z_i'-m_Z(z)\bigr)^2=v_Z(z),
		\label{eq:trinomial-moment-matching}
	\end{equation}
	where $Z_i'$ are the selected successor values. Thus both exact marginal CIR
	moments are matched at every active node.
	
	The scale restriction $0<\lambda_Z<2$ has a direct positivity
	interpretation. If the conditional mean $m$ lies between adjacent grid values
	$z_\ell<m<z_u$, the smallest variance attainable by a grid-supported
	distribution with mean $m$ is $(m-z_\ell)(z_u-m)$. Away from the lower
	boundary, the worst-case ratio of this minimum to the leading CIR conditional
	variance is asymptotically $\lambda_Z^2/4$. Values of $\lambda_Z$ close to
	two therefore leave little probability-positivity margin, whereas smaller
	values increase the number of factor nodes.
	
	Because the Heston variance factor $V$ and the CIR rate factor $X$ underlying
	the CIR++ specification have markedly different mean-reversion dynamics, their
	grid scales are allowed to differ. A numerical screening of candidate pairs
	selected $(\lambda_V,\lambda_X)=(1,1.4)$ for the central parameterization.
	Among the screened configurations, this pair provided the best end-to-end
	pricing accuracy while maintaining exact non-negative moment matching, strict
	local coupling, and a moderate joint state-space size.
	
	For a parent state $z$ and a selected successor $Z_i'$, define the centered
	variance-one marginal innovation by
	\[
	z_i^Z
	=
	\frac{Z_i'-m_Z(z)}{\sqrt{v_Z(z)}}.
	\]
For trinomial marginal probabilities $p_i^V,p_j^X$, the strict
implementation uses the one-parameter covariance-preserving joint table
\begin{equation}
	\pi_{ij}
	=
	p_i^Vp_j^X
	\bigl(1+\rho_{Vr}z_i^Vz_j^X\bigr).
	\label{eq:trinomial-coupling}
\end{equation}
Because the marginal innovations are centered and have unit variance, this
construction preserves both marginal transition laws and reproduces the target
innovation correlation $\rho_{Vr}$. Since $p_i^V,p_j^X\ge0$, the coupling is
strictly admissible whenever
$1+\rho_{Vr}z_i^Vz_j^X>0$ for every active pair $(i,j)$. 
For $(\lambda_V,\lambda_X)=(1,1.4)$, a sufficient uniform interior range is
$|\rho_{Vr}|<5/14$. This range is broad relative to commonly used
variance--rate correlations, which are often set to zero in hybrid
stochastic-volatility/stochastic-rate models; see, for example,
\cite{KangZiveyi2018}. This is a sufficient range for the one-parameter covariance-preserving
construction, not necessarily the widest covariance interval attainable
among all $3\times3$ tables.
	
	The stated interior range follows from the limiting marginal stencils. The
	next lemma records the calculation because strict local admissibility is part
	of the convergence argument, not merely an implementation diagnostic.
	
\begin{lemma}
	\label{lem:trinomial-interior-probabilities-main}
	Consider a CIR factor $Z$ on the square-root grid
	$Z_q=(\sqrt{Z_0}+q\lambda\sigma\sqrt h/2)^2$ and let
	$K_Z$ be a compact subset of $(0,\infty)$. For the symmetric stencil
	$(q-1,q,q+1)$, denote by $p_-(z,h)$, $p_0(z,h)$, and $p_+(z,h)$
	the transition probabilities to $Z_{q-1}$, $Z_q$, and $Z_{q+1}$,
	respectively. Uniformly for $z\in K_Z$:
		\begin{enumerate}[label=(\roman*)]
			\item if $\lambda>1$, the symmetric stencil $(q-1,q,q+1)$ is
			admissible for all sufficiently small $h$, with
			\begin{align} 
				p_0(z,h)
				&=
				1-\frac1{\lambda^2}+O(h);
				\label{eq:tri-central-lambda-gt-one}\\
				p_\pm(z,h)
				&=
				\frac{1}{2\lambda^2}
				\pm
				\frac{\kappa(\theta-z)-\sigma^2/4}
				{2\lambda\sigma\sqrt z}\sqrt h
				+O(h),
				\label{eq:tri-prob-lambda-gt-one} 
			\end{align}
			
			\item if $\lambda=1$, define
			\begin{equation}
				A(z)
				=
				\kappa
				-
				\frac{\kappa^2(\theta-z)^2}{\sigma^2z}
				-
				\frac{\sigma^2}{16z}.
				\label{eq:tri-A-function}
			\end{equation}
			When $A(z)>0$, the symmetric stencil is admissible and
			\begin{align}
				p_0(z,h)
				&=
				A(z)h+O(h^2),
				\label{eq:tri-lambda-one-central}
				\\
				p_\pm(z,h)
				&=
				\frac12
				\pm
				\frac{\kappa(\theta-z)-\sigma^2/4}
				{2\sigma\sqrt z}\sqrt h
				-\frac12A(z)h
				+O(h^{3/2}).
				\label{eq:tri-lambda-one-side}
			\end{align}
			When $A(z)<0$, the minimal-span search selects a one-sided
			width-three stencil whose limiting probabilities are, up to reversal,
			$(1/6,1/2,1/3)$.
		\end{enumerate}
	\end{lemma}
	
	\noindent\emph{Proof.} See Appendix~\ref{app:proofs}.
	
	\begin{proposition}
		\label{prop:trinomial-rho-range-main}
		Let $K$ be a compact subset of $(0,\infty)\times(0,\infty)$. For the
		adaptive-trinomial chain with
		$(\lambda_V,\lambda_X)=(1,1.4)$, every correlation satisfying $
			|\rho_{Vr}|<\frac5{14}$
		is strictly admissible on $K$ for all sufficiently small $h$.
	\end{proposition}
	
	\noindent\emph{Proof.} See Appendix~\ref{app:proofs}.
	
	The bound $|\rho_{Vr}|<5/14$ corresponds to the worst-case enlarged
	variance stencil, where $\max_{i,j}|z_i^Vz_j^X|$ tends to
	$2\times1.4=2.8$. If the variance stencil remains symmetric, the same
	quantity tends to $1\times1.4=1.4$, and the local admissible range widens
	to $|\rho_{Vr}|\le5/7$.
	
	Both joint factor constructions are combined with the conditional fund
	transition described in the next subsection.
	
\subsection{Correlated fund transition}

The following construction is common to the binomial and adaptive-trinomial
factor lattices. At each time discretization, the deterministic CIR++ shift
is calibrated on the discrete rate lattice so that the initial market
discount curve is reproduced exactly at all grid maturities; details are
given in the Supplementary Material.

At time $t_n$ and parent node $(j,k)$, let $b$ index a
positive-probability joint factor branch. Its probability is $p_{n,b}$ and
its child indices are $(j_b,k_b)$. The branchwise integrated short rate is
approximated by
\[
R_{n,b}
=
\frac{h}{2}\bigl(x_k+x_{k_b}\bigr)
+
\int_{t_n}^{t_{n+1}}\varphi(u)\,du .
\]
The one-step fund transition uses the parent variance $v_j$.

In what follows, $\rho_{Vr,h}^{\rm eff}$ denotes the variance--rate
correlation actually used by the local joint coupling. For a strictly
admissible coupling, $\rho_{Vr,h}^{\rm eff}=\rho_{Vr}$; for the projected
binomial fallback, it is the projected value defined in
\eqref{eq:effective-rho}. We then set
\begin{equation}
	\alpha_V
	=
	\frac{\rho_{SV}-\rho_{Sr}\rho_{Vr,h}^{\rm eff}}
	{1-(\rho_{Vr,h}^{\rm eff})^2},
	\quad
	\alpha_r
	=
	\frac{\rho_{Sr}-\rho_{SV}\rho_{Vr,h}^{\rm eff}}
	{1-(\rho_{Vr,h}^{\rm eff})^2},
	\quad
	\alpha_\perp
	=
	\left[
	1-
	\frac{
		\rho_{SV}^2
		-2\rho_{SV}\rho_{Sr}\rho_{Vr,h}^{\rm eff}
		+\rho_{Sr}^2
	}{
		1-(\rho_{Vr,h}^{\rm eff})^2
	}
	\right]^{1/2}
	\ge0.
	\label{eq:equity-decomposition}
\end{equation}
	If $\eta$ is independent with mean zero and variance one, then
	$\xi_S=\alpha_V\xi_V+\alpha_r\xi_r+\alpha_\perp\eta$ has unit variance and the
	required equity correlations. For branch $b$ and quadrature point $\ell$,
		\begin{align}
		F_{n+1}&=(F_n+D_n)M_{n,b,\ell},\notag\\
		M_{n,b,\ell}
		&=
		\exp\!\left\{
		R_{n,b}-qh-\tfrac12 v_jh
		+\sqrt{v_jh}
		\bigl(
		\alpha_V\xi_{V,b}
		+\alpha_r\xi_{r,b}
		+\alpha_\perp\eta_\ell
		\bigr)
		\right\}.
		\label{eq:fund-transition}
	\end{align} 
	The binomial comparison uses the two-point residual rule
	$\eta\in\{-1,1\}$ with weights $(1/2,1/2)$, yielding at most
	$8$ hybrid branches per factor node. The production adaptive-trinomial
	scheme uses
	$\eta\in\{-\sqrt3,0,\sqrt3\}$ with weights $(1/6,2/3,1/6)$,
	yielding at most $27$ hybrid branches.
	
	A local multiplicative martingale correction is applied to the fund
	multipliers. If $M_{n,b,\ell}^{(0)}$ denotes the uncorrected multiplier, set
	\[
	\Lambda_{n,j,k}
	=
	\frac{e^{-qh}}
	{\displaystyle
		\sum_{b,\ell}
		p_{n,b,\ell}e^{-R_{n,b}}M_{n,b,\ell}^{(0)}},
	\qquad
	M_{n,b,\ell}
	=
	\Lambda_{n,j,k}M_{n,b,\ell}^{(0)}.
	\]
	Then
	\begin{equation}
		\sum_{b,\ell}
		p_{n,b,\ell}e^{-R_{n,b}}M_{n,b,\ell}
		=
		e^{-qh},
		\label{eq:martingale-correction}
	\end{equation}
	so that
	$\mathbb E_n^h[e^{-R_{n,n+1}^h}F_{n+1}]
	=(F_n+D_n)e^{-qh}$ exactly.

\subsection{Backward recursion and adaptive fund representation}

Let $\omega=(b,\ell)$ index a positive-probability hybrid branch, and set
$p_{n,\omega}=p_{n,b,\ell}$, $R_{n,\omega}=R_{n,b}$,
$M_{n,\omega}=M_{n,b,\ell}$, and
$(j_\omega,k_\omega)=(j_b,k_b)$. The set of such branches from node
$(j,k)$ is denoted by $\Omega_n(j,k)$.

At each joint factor node $(v_j,x_k)$, the value function is represented by
a sorted set of fund knots and a piecewise polynomial. Candidate knots are
generated by inverse images of child knots under the affine branch maps
$F\mapsto(F+D_n)M_{n,\omega}$, together with payoff kinks,
continuation--surrender intersections, and domain endpoints.

On the last numerical step of policy year $i$, characterized by
$t_{n+1}=i+1$, define the mortality-adjusted child value by
\begin{equation}
	U_{n+1}^{\mathrm{mort}}(F,v,x;P)
	=
	(1-q_{a_0+i})U_{n+1}(F,v,x;P)
	+
	q_{a_0+i}B_{i+1}(F).
	\label{eq:mortality-map}
\end{equation}
On all other numerical steps, set
$U_{n+1}^{\mathrm{mort}}=U_{n+1}$.

The continuation operator is
\begin{equation}
	\begin{aligned}
		C_n(F,v_j,x_k;P)
		={}&
		-P\,\mathbf 1_{\{t_n\in\{0,\ldots,T-1\}\}}
		\\
		&+
		\sum_{\omega\in\Omega_n(j,k)}
		p_{n,\omega}e^{-R_{n,\omega}}
		U_{n+1}^{\mathrm{mort}}
		\bigl((F+D_n)M_{n,\omega},
		v_{j_\omega},x_{k_\omega};P\bigr).
	\end{aligned}
	\label{eq:continuation-operator}
\end{equation}

On each retained fund interval $[a,b]$, define the four local interpolation
points
\begin{equation}
	f_0=a,\qquad
	f_1=a+\left(1-\frac1{\sqrt2}\right)(b-a),\qquad
	f_2=a+\frac1{\sqrt2}(b-a),\qquad
	f_3=b.
	\label{eq:cubic-interpolation-points}
\end{equation}
The production scheme represents the value function on $[a,b]$ by the unique
cubic interpolant through $f_0,f_1,f_2,f_3$.
These symmetric locations minimize
$\max_{0\le t\le1}|t(t-\alpha)(t-\beta)(t-1)|$ within the symmetric
two-point family, and hence minimize the standard fourth-order interpolation
error bound on a smooth interval of given width.

At surrender dates, continuation--payoff intersections are inserted as knots
before applying \eqref{eq:surrender-obstacle}. Exercise intervals are then
compressed exactly to their endpoints. The terminal payoff kink and affine
right tail are retained, so the fund cap limits only the representation domain
and does not clip values.

\begin{remark}
	The resulting interpolating cubic is not, in general, the best cubic
	approximation to the value function under the local error criterion used for
	pruning. Its coefficients are, however, obtained directly from four function
	values, without solving a local optimization problem. Since this construction
	is repeated across fund intervals, factor states, and time steps, its simplicity
	provides a favorable accuracy--cost trade-off for the overall algorithm.
\end{remark}
	
\subsection{Certified pruning, Delta control, and complexity}
\label{subsec:pruning-complexity}

Backward propagation maps each retained child knot through every hybrid
branch, so the number of candidate singular points may grow exponentially
with the number of time steps in the worst case. The production scheme
controls this growth by applying a certified Douglas--Peucker-type pruning
\cite{DouglasPeucker1973} to the piecewise-cubic value functions.

On each candidate fund interval $[a,b]$, let $U(F)$ denote the current
piecewise-cubic value function and let $\widetilde U(F)$ denote its pruned
cubic approximation. Since $U-\widetilde U$ is cubic on each common
subinterval, its maximum absolute deviation is obtained by checking the
endpoints and all interior stationary points, which are found by solving a
quadratic equation. If the prescribed tolerance is violated, the interval is
split at a point where this maximum is attained and the test is repeated
recursively. The retained approximation therefore satisfies, for every
$F\in[a,b]$,
\begin{equation}
	|U(F)-\widetilde U(F)|
	\le
	\varepsilon_{\rm abs}+\varepsilon_{\rm rel}F,
	\label{eq:pruning-certificate}
\end{equation}
where $\varepsilon_{\rm abs}$ and $\varepsilon_{\rm rel}$ are the prescribed
absolute and relative pruning tolerances, respectively. A conservative
interval envelope based on the minimum fund level over each common
subinterval guarantees this mixed tolerance continuously. Optional
probability-weighted, preimage, and emergency safeguards are not used in the
production results.

\begin{proposition}\label{prop:delta-bound}
	On a common-refinement cell $I=[a,b]$ containing no preserved
	singularity,
	\begin{equation}
		\|\partial_FU-\partial_F\widetilde U\|_{\infty,I}
		\le
		\frac{18}{b-a}
		\|U-\widetilde U\|_{\infty,I}.
		\label{eq:delta-bound}
	\end{equation}
\end{proposition}

\begin{proof}
	The difference $U-\widetilde U$ is a polynomial of degree at most three on
	$I$. Rescaling $I$ to $[-1,1]$ and applying Markov's inequality
	(see \cite{Rivlin1990}) yields the factor $2\times3^2=18$.
\end{proof}

Thus, away from preserved singularities, value-function pruning also controls
the local Delta error, with the bound scaling inversely with the interval
width. End-to-end Delta accuracy is tested numerically in
Section~\ref{subsec:delta-validation}.

At time level $n$, each marginal factor lattice has $O(n)$ nodes and the joint
grid has $O(n^2)$ node pairs. If $N^{\rm knot}_{n,j,k}$ knots are retained at
node $(j,k)$, permanent memory at that level is proportional to
$\sum_{j,k}N^{\rm knot}_{n,j,k}$. The main temporary cost arises from merging
the branch preimages. The implementation uses a streaming $k$-way merge and
OpenMP parallelism across joint factor nodes.
	
	\section{Convergence, Delta consistency, and Richardson extrapolation}
	\label{sec:convergence}
	
	The results are organized in six steps: (i) local consistency and identification
	of the limiting generator; (ii) weak convergence of the strict binomial and
	adaptive-trinomial chains, with convergence of European and finite-date Bermudan
	values; (iii) convergence of the projected binomial fallback and stability of
	the contractual recursion; (iv) convergence of the adaptive singular-point
	representation, including pruning and fair-premium determination; (v) under
	stronger smoothness assumptions, first-order weak accuracy for the strict,
	unpruned binomial scheme and regular-region Delta consistency; and (vi) the
	stronger asymptotic expansion required for Richardson extrapolation. No global
	first-order rate is claimed for the production adaptive-trinomial scheme. The
	surrender feature studied here leads to a finite-date Bermudan stopping problem,
	whereas continuous-exercise American stopping is outside the present theoretical
	scope.
	
	Between contractual dates set $Y_t=\log F_t$ and
	$\mathbf Z_t=(Y_t,V_t,X_t)$. The unpruned strict-chain value at time $t_n$ is
	$U_n^h$, and its adaptive representation is $\widetilde U_n^h$. For a scalar
	valuation problem, $\mathcal V$ denotes the exact value and $\mathcal V_h$ the
	unpruned time-discrete value. On a regular fund interval, Delta denotes the
	fund derivative $\partial_F U$; at a preserved kink, left and right derivatives
	are considered separately.
	
	Unless stated otherwise, all local $O(\cdot)$ and $o(\cdot)$ estimates are
	uniform on compact subsets of $\mathbb R\times(0,\infty)\times(0,\infty)$. For the rate
	statements, the market curve is assumed sufficiently smooth and the calibrated
	shift increments satisfy
	\begin{equation}
		s_{n,h}=\int_{t_n}^{t_{n+1}}\varphi(u)\,\mathrm du+O(h^2)
		\label{eq:shift-local-consistency}
	\end{equation}
	uniformly on the time grid.  
	
	When the martingale correction is applied, its
	parent-node rescaling factor $\Lambda_{n,j,k}$ satisfies
	$\Lambda_{n,j,k}=1+o(h)$ under the local-consistency assumptions, so it does
	not change the limiting drift or covariance. For the binomial weak-order and
	Richardson results, the stronger expansion $\Lambda_{n,j,k}=1+O(h^2)$ is
	assumed and follows from the corresponding integer-power one-step expansion.
	The stronger rate statements impose the additional smoothness and operator
	expansions given below.
	
	All results concern a sequence of grids with $h=T/N\downarrow0$ whose
	contractual dates belong to every grid. Unless stated otherwise, the financial
	chain is strict, unprojected, and unpruned; projection and adaptive
	representation are introduced only after convergence of the strict chain has
	been established. Selected technical proofs are deferred to
	Appendix~\ref{app:proofs}.
	
	Throughout this section, $O(h^\alpha)$ denotes a quantity whose absolute
	value is bounded by $C h^\alpha$ for all sufficiently small $h$, with a
	constant $C$ independent of $h$. When the estimate is stated on a compact
	set, the bound is understood to hold uniformly on that set.
	
	\subsection{Local consistency and the limiting generator}
	\label{subsec:local-consistency}
	
	Let $(Y_n^h,V_n^h,X_n^h)$ be the chain generated by the two factor lattices and the
	conditional fund quadrature. Uniformly on compact subsets of
	$\mathbb R\times(0,\infty)\times(0,\infty)$, the marginal moment properties
	of the two lattice constructions, together with strict covariance matching,
	give
	\begin{align}
		\mathbb E[\Delta V_n^h\mid v,x]
		&=\kappa_V(\theta_V-v)h+o(h),
		&\mathbb E[(\Delta V_n^h)^2\mid v,x]
		&=\sigma_V^2vh+o(h),\label{eq:local-v-moments}\\
		\mathbb E[\Delta X_n^h\mid v,x]
		&=\kappa_r(\theta_r-x)h+o(h),
		&\mathbb E[(\Delta X_n^h)^2\mid v,x]
		&=\sigma_r^2xh+o(h),\label{eq:local-x-moments}\\
		\mathbb E[\Delta V_n^h\Delta X_n^h\mid v,x]
		&=\rho_{Vr}\sigma_V\sigma_r\sqrt{vx}\,h+o(h).
		\label{eq:local-vx-moment}
	\end{align}
	The equity decomposition in \eqref{eq:equity-decomposition}, the symmetric
	residual quadrature, and the local martingale correction yield
	\begin{align}
		\mathbb E[\Delta Y_n^h\mid y,v,x]
		&=\left(x+\varphi(t_n)-q-\tfrac12v\right)h+o(h),
		\label{eq:local-y-mean}\\
		\mathbb E[(\Delta Y_n^h)^2\mid y,v,x]
		&=vh+o(h),\label{eq:local-y-variance}\\
		\mathbb E[\Delta Y_n^h\Delta V_n^h\mid y,v,x]
		&=\rho_{SV}\sigma_Vv\,h+o(h),\label{eq:local-yv-moment}\\
		\mathbb E[\Delta Y_n^h\Delta X_n^h\mid y,v,x]
		&=\rho_{Sr}\sigma_r\sqrt{vx}\,h+o(h).
		\label{eq:local-yx-moment}
	\end{align}
	Consequently, the discrete generator converges to
	\begin{align}
		\mathcal L_t\psi={}&
		\left(x+\varphi(t)-q-\tfrac12v\right)\partial_y\psi
		+\kappa_V(\theta_V-v)\partial_v\psi
		+\kappa_r(\theta_r-x)\partial_x\psi \notag\\
		&+\tfrac12v\partial_{yy}\psi
		+\tfrac12\sigma_V^2v\partial_{vv}\psi
		+\tfrac12\sigma_r^2x\partial_{xx}\psi \notag\\
		&+\rho_{SV}\sigma_Vv\partial_{yv}\psi
		+\rho_{Sr}\sigma_r\sqrt{vx}\partial_{yx}\psi
		+\rho_{Vr}\sigma_V\sigma_r\sqrt{vx}\partial_{vx}\psi.
		\label{eq:limiting-generator}
	\end{align}
	The discounted pricing generator is
	$\mathcal A_t\psi=\mathcal L_t\psi-[x+\varphi(t)]\psi$.
	The vanishing-jump property follows because every active successor lies
	$O(\!\sqrt h)$ from its parent on an interior compact set.
	
	\subsection{Weak convergence of the strict chains and contract values}
	\label{subsec:weak-convergence}
	
	The local consistency relations above, together with tightness and vanishing
	jumps, identify the diffusion limit. The same argument also gives convergence
	of discounted contract values when the payoff family is uniformly integrable.
	
	\begin{theorem}
		\label{thm:strict-chain-weak-convergence}
		Assume local consistency in
		\eqref{eq:local-v-moments}--\eqref{eq:local-yx-moment}, vanishing jumps, compact
		containment, uniform polynomial moment bounds, and uniqueness of the martingale
		problem for the Heston--CIR++ diffusion. Then both the strict binomial and the
		strict adaptive-trinomial chains converge weakly to the diffusion on every
		finite horizon. If discounted payoffs have linear growth and are uniformly
		integrable, the corresponding European and finite-date Bermudan values
		converge.
	\end{theorem}
	
	\begin{proof}
		The first two conditional moments converge locally to the drift and covariance
		of \eqref{eq:limiting-generator}, and all jumps vanish uniformly on interior
		compact sets. Compact containment prevents loss of mass at the localization
		boundary. The Markov-chain approximation theorem of
		Kushner and Dupuis \cite{KushnerDupuis2001} therefore identifies every weak
		subsequential limit with a solution of the Heston--CIR++ martingale problem;
		uniqueness yields convergence of the full sequence. Uniform integrability
		allows discounted expectations of linear-growth payoffs to pass to the limit.
		A finite collection of contractual and exercise dates is handled by backward
		induction through the continuous event maps.
	\end{proof}
	
	The European case corresponds to valuation without early exercise. With annual
	surrender dates, the contract instead gives a finite-date Bermudan stopping
	problem, which is covered by the final statement of
	Theorem~\ref{thm:strict-chain-weak-convergence}.
	
	\subsection{Projected binomial chain and stability of the contractual recursion}
	\label{subsec:projection-stability}
	
	For the binomial marginal defined in
	Section~\ref{subsubsec:binomial-factor-lattices}, the exact-mean interpolation
	probability satisfies, uniformly on interior compact sets,
	\begin{equation}
		p_{Z,h}(z)
		=
		\frac12+\beta_Z(z)\sqrt h+O(h^{3/2}),
		\qquad
		\beta_Z(z)
		=
		\frac{\kappa_Z(\theta_Z-z)-\sigma_Z^2/4}
		{2\sigma_Z\sqrt z}.
		\label{eq:binomial-prob-expansion}
	\end{equation}
	
	For marginal up probabilities $p_V$ and $p_X$, the attainable correlation
	interval of the $2\times2$ joint table has endpoints
	\begin{align}
		\rho_+(p_V,p_X)
		&=
		\exp\{-\tfrac12|\ell(p_V)-\ell(p_X)|\},
		\notag\\
		\rho_-(p_V,p_X)
		&=
		-\exp\{-\tfrac12|\ell(p_V)+\ell(p_X)|\},
		\qquad
		\ell(u)=\log\frac{u}{1-u}.
		\label{eq:binomial-correlation-endpoints}
	\end{align}
	Consequently, on every compact subset of $(0,\infty)\times(0,\infty)$,
	the corresponding local endpoints satisfy
	$\rho_-(v,x;h)\to-1$ and $\rho_+(v,x;h)\to1$ uniformly as $h\downarrow0$.
	
	\begin{proposition}
		\label{prop:uniform-interior-admissibility}
		Let $K$ be a compact subset of $(0,\infty)\times(0,\infty)$.
		For every fixed $\rho_{Vr}\in(-1,1)$, the unprojected
		$2\times2$ joint table is strictly admissible on $K$ for all
		sufficiently small $h$.
	\end{proposition}
	
	\noindent\emph{Proof.} See Appendix~\ref{app:proofs}.
	
	Let $A_h$ be the set of factor states at which the requested $\rho_{Vr}$
	lies outside the local binomial admissibility interval, and define the first
	projected time by
	\begin{equation}
		\tau_h
		=
		\inf\{t_n\le T:(V_n^h,X_n^h)\in A_h\}.
		\label{eq:first-projected-time}
	\end{equation}
	
	\begin{proposition}
		\label{prop:vanishing-projection}
		Assume the strict Feller conditions
		$2\kappa_V\theta_V>\sigma_V^2$ and
		$2\kappa_r\theta_r>\sigma_r^2$. Suppose the two marginal CIR lattices are
		tight on $[0,T]$ and converge weakly to their CIR limits. Since projection
		preserves both marginals, the same compact-containment property holds for the
		projected joint chain. Then
		\begin{equation}
			\mathbb P(\tau_h\le T)\longrightarrow0.
			\label{eq:vanishing-projection}
		\end{equation}
	\end{proposition}
	
	\noindent\emph{Proof.} See Appendix~\ref{app:proofs}.
	
	Define the time-weighted occupation error
	\begin{equation}
		\mathcal M_h=h\sum_{n:t_n<T}
		\mathbb E\!\left[
		|\rho_{Vr,h}^{\rm eff}-\rho_{Vr}|\,
		\mathbf 1_{A_h}(V_n^h,X_n^h)
		\right].
		\label{eq:projection-occupation}
	\end{equation}
	Since correlations are bounded by one,
	$\mathcal M_h\le2T\mathbb P(\tau_h\le T)\to0$.
	
	\begin{theorem}
		\label{thm:projected-chain}
		Assume Proposition~\ref{prop:vanishing-projection}, local consistency of the
		unprojected chain on interior compact sets, vanishing jumps, uniform moment
		bounds, and uniqueness of the martingale problem for the Heston--CIR++
		diffusion. Assume also that every projected local value
		$\rho_{Vr,h}^{\rm eff}$ leaves the residual equity variance in
		\eqref{eq:equity-decomposition} non-negative, so the three-factor transition
		is well defined. Then the projected binomial chain converges weakly to the same
		limit as the strict localized chain. If discounted contract payoffs are
		uniformly integrable, European and finite-date Bermudan values converge as
		well.
	\end{theorem}
	
	\noindent\emph{Proof.} See Appendix~\ref{app:proofs}.
	
	Because the CIR++ shift can be negative, the discounted continuation operator
	need not be a contraction in the ordinary supremum norm. Define
	$\varphi_-:=\max_{0\le t\le T}(-\varphi(t))_+$. The branchwise rate
	integral is the sum of a non-negative CIR contribution and a calibrated shift
	increment; by \eqref{eq:shift-local-consistency},
	$e^{-R_{n,b}}\le e^{\varphi_-h+O(h^2)}$ uniformly on the grid. Hence one
	numerical step is stable with factor $1+O(h)$ and the product of all stability
	factors is bounded uniformly over the finite horizon.
	Contribution translations are isometries, premium updates are affine, the
	mortality map is a convex combination, and the obstacle is non-expansive:
	\begin{equation}
		|\max\{f,R\}-\max\{g,R\}|\le|f-g|.
		\label{eq:obstacle-nonexpansive}
	\end{equation}
	The complete finite-horizon recursion is therefore stable.
	
	\subsection{Adaptive representation and fair-premium convergence}
	\label{subsec:adaptive-convergence}
	
	For $J\in\{\mathrm{NS},\mathrm{FB},\mathrm{GB},\mathrm{MX}\}$, let $\Phi_J(P)$ and $\Phi_J^h(P)$ denote the exact and
	numerical initial net values at annual premium $P$, and let $P_J$ and $P_J^h$
	be their respective roots.
	
	Let $U_n^h$ be the unpruned value of the discrete chain and
	$\widetilde U_n^h$ its adaptive approximation. Suppose ordinary interpolation
	and pruning errors at step $n$ are bounded by $\iota_{n,h}$ and
	$\varepsilon_{n,h}$ on the represented fund interval. Stability gives
	\begin{equation}
		\|\widetilde U_0^h-U_0^h\|_\infty
		\le C_T\sum_{n=0}^{N-1}(\iota_{n,h}+\varepsilon_{n,h})
		+\mathcal E_h^{\rm bdry},
		\label{eq:global-representation-bound}
	\end{equation}
	where $C_T$ is independent of $h$ and $\mathcal E_h^{\rm bdry}$ is the
	remaining factor- and fund-boundary error. For the implemented payoffs, the
	right tail is exactly affine and stored through its slope; once the cap lies
	beyond all non-affine singularities it creates no value truncation. On a
	quasi-uniform time grid, uniform per-step representation error $o(h)$ is a
	sufficient, though not necessary, condition for the sum in
	\eqref{eq:global-representation-bound} to vanish.
	
	\begin{theorem}
		\label{thm:adaptive-convergence}
		Assume the strict financial chain is locally consistent, tight, and uniformly
		integrable; the factor and fund boundary errors vanish; and the right-hand side
		of \eqref{eq:global-representation-bound} tends to zero. Then the adaptive
		singular-point value converges to the Heston--CIR++ contract value with the
		prescribed finite collection of premium, mortality, and surrender dates.
	\end{theorem}
	
	\begin{proof}
		Local consistency, tightness, and vanishing jumps imply weak convergence of the
		unpruned chain. Uniform integrability allows discounted expectations of the
		linear-growth payoffs to pass to the limit. Backward application of the finite
		collection of stable translations, premium updates, mortality maps, and
		non-expansive obstacle operators preserves convergence. The remaining
		difference between the unpruned discrete value and its adaptive representation
		is bounded by \eqref{eq:global-representation-bound}, which vanishes by
		assumption.
	\end{proof}
	
	\begin{corollary}
		\label{cor:projected-adaptive}
		Under the hypotheses of Theorem~\ref{thm:projected-chain}, if the boundary and
		representation errors in \eqref{eq:global-representation-bound} vanish, the
		projected adaptive binomial scheme converges to the same contract value.
	\end{corollary}
	
	Define
	\begin{equation}
		{}_ip_{a_0}=\prod_{u=0}^{i-1}(1-q_{a_0+u}),
		\qquad {}_0p_{a_0}=1,
		\label{eq:survival-to-premium-date}
	\end{equation}
	as the probability of surviving to anniversary $i$. Because premiums are paid
	only while the insured is alive, the no-surrender discrete net value is exactly
	affine:
	\begin{equation}
		\Phi_{\mathrm{NS}}^h(P)=\mathcal B_h-PA^M,
		\qquad
		A^M=\sum_{i=0}^{T-1}{}_ip_{a_0}\,P^M(0,i),
		\qquad
		P_{\mathrm{NS}}^h=\mathcal B_h/A^M.
		\label{eq:affine-fair-premium}
	\end{equation}
	The deterministic survival probabilities and the discrete CIR++ calibration make $A^M$ grid independent. Surrender
	premiums are obtained by safeguarded Brent--Dekker root finding
	\cite{Brent1973}.
	
	\begin{proposition}
		\label{prop:fair-premium-convergence}
		Suppose $\Phi_J^h\to\Phi_J$ uniformly on a compact premium interval,
		$\Phi_J$ has a unique interior root $P_J$, and
		$|\Phi_J'(P)|\ge m_J>0$ near $P_J$. Then $P_J^h\to P_J$ and
		\begin{equation}
			|P_J^h-P_J|\le
			\frac{\sup_P|\Phi_J^h(P)-\Phi_J(P)|}{m_J}
			+o\!\left(\sup_P|\Phi_J^h(P)-\Phi_J(P)|\right).
			\label{eq:fair-premium-root-bound}
		\end{equation}
	\end{proposition}
	
	\begin{proof}
		Uniform convergence implies that, for sufficiently small $h$, a numerical root
		lies in the neighborhood where the derivative is bounded away from zero.
		Since $\Phi_J(P_J)=0=\Phi_J^h(P_J^h)$,
		\begin{align*}
			|\Phi_J(P_J^h)-\Phi_J(P_J)|
			&=|\Phi_J(P_J^h)-\Phi_J^h(P_J^h)|\\
			&\le\sup_P|\Phi_J^h(P)-\Phi_J(P)|.
		\end{align*}
		The mean-value theorem and $|\Phi_J'|\ge m_J$ give the leading bound. The
		stated remainder follows from the local expansion of the inverse map at the
		simple root.
	\end{proof}
	
	\subsection{Weak-error rates}
	\label{subsec:weak-order}
	
	The preceding results establish convergence of the valuation method. A
	first-order weak-error bound requires additional moment structure and smoothness
	and is therefore stated separately.
	
	The binomial scheme has the stronger moment structure needed for a first-order
	rate. Let $\Delta Z_h^{\rm bin}$ denote one strict binomial increment.
	
	\begin{lemma}
		\label{lem:binomial-scaled-moments}
		Uniformly on interior compact sets,
		\begin{align}
			\mathbb E[\Delta Z_{h,i}^{\rm bin}]
			&=b_i(t,Z)h+O(h^2),\label{eq:bin-scaled-first}\\
			\mathbb E[\Delta Z_{h,i}^{\rm bin}\Delta Z_{h,j}^{\rm bin}]
			&=a_{ij}(t,Z)h+O(h^2),\label{eq:bin-scaled-second}\\
			\mathbb E[\Delta Z_{h,i}^{\rm bin}\Delta Z_{h,j}^{\rm bin}
			\Delta Z_{h,k}^{\rm bin}]&=O(h^2),\label{eq:bin-scaled-third}\\
			\mathbb E[|\Delta Z_{h,i}^{\rm bin}\Delta Z_{h,j}^{\rm bin}
			\Delta Z_{h,k}^{\rm bin}\Delta Z_{h,\ell}^{\rm bin}|]&=O(h^2),
			\label{eq:bin-scaled-fourth}
		\end{align}
		and $\mathbb E|\Delta Z_h^{\rm bin}|^5=O(h^{5/2})$.
	\end{lemma}
	
	\noindent\emph{Proof.} See Appendix~\ref{app:proofs}.
	
	\begin{theorem}
		\label{thm:binomial-weak-order}
		Suppose that the strict binomial chain satisfies the local moment estimates
		in Lemma~\ref{lem:binomial-scaled-moments}. Assume also uniform polynomial
		moment bounds, stability of the discounted operators, and a backward value
		with five state derivatives of polynomial growth between contractual dates.
		If the finitely many deterministic event maps preserve this regularity, then
		\begin{equation}
			|\mathcal V_h^{\rm bin}-\mathcal V|
			\le C_T h=\frac{C_TT}{N}.
			\label{eq:binomial-weak-order}
		\end{equation}
	\end{theorem}
	
	\noindent\emph{Proof.} See Appendix~\ref{app:proofs}.
	
	The proof is based on the one-step expansion, valid for smooth test functions,
	\begin{equation}
		Q_{t,h}^{\rm bin}f=f+h\mathcal A_tf+O(h^2).
		\label{eq:binomial-local-expansion}
	\end{equation}
	Its derivation from Lemma~\ref{lem:binomial-scaled-moments} is given in
	Appendix~\ref{app:proofs}.
	
	For an adaptive trinomial with $\lambda>1$, the limiting symmetric marginal
	probabilities in Lemma~\ref{lem:trinomial-interior-probabilities-main} imply the
	same $O(\!\sqrt h)$ third standardized moments, so the preceding first-order
	argument applies on localized regions where all selected stencils are
	eventually symmetric. The production variance grid uses $\lambda_V=1$ and
	requires a separate qualification.
	
	\begin{proposition}
		\label{prop:trinomial-one-sided-moments}
		At an interior variance state where $A(v)<0$ in
		\eqref{eq:tri-A-function}, the limiting standardized support of the selected
		minimal-span stencil is, up to reflection, $(-2,0,1)$ with probabilities
		$(1/6,1/2,1/3)$. Its third standardized moment is $-1$ (or $+1$ after
		reflection). Consequently, the corresponding raw third moment of the
		variance increment is generically $O(h^{3/2})$, not $O(h^2)$.
	\end{proposition}
	
	\begin{proof}
		The limiting support and probabilities follow from
		Lemma~\ref{lem:trinomial-interior-probabilities-main}. Direct calculation
		gives
		\begin{equation*}
			\frac16(-2)^3+\frac12 0^3+\frac13 1^3=-1.
		\end{equation*}
		Multiplication by the cube of the $O(\!\sqrt h)$ local scale gives an
		$O(h^{3/2})$ raw third moment.
	\end{proof}
	
	\begin{corollary}
		\label{cor:trinomial-rate-scope}
		The first-order weak bound in
		Theorem~\ref{thm:binomial-weak-order} extends to an adaptive-trinomial
		refinement only under an eventual symmetric-stencil condition, for example on
		a localized region where $A_V(v)\ge c_A>0$ for some constant $c_A$, where
		$A_V$ denotes the function in \eqref{eq:tri-A-function} evaluated with
		$(\kappa_V,\theta_V,\sigma_V)$, or when both grid multipliers exceed one.
		For the production choice $\lambda_V=1$, strict weak convergence follows
		from Theorem~\ref{thm:strict-chain-weak-convergence}, but the moment argument
		above does not establish a global first-order rate.
	\end{corollary}
	
	\begin{remark}
		\label{rem:trinomial-rate-interpretation}
		The one-sided-stencil issue is not a negligible tail event for the baseline
		variance process; it is a structural feature of the $\lambda_V=1$ grid. We
		therefore report direct trinomial refinements without fitting or claiming a
		formal convergence order. The numerical sequence is used as an accurate
		strict approximation, while the binomial chain supplies the proved
		first-order/Richardson cross-check.
	\end{remark}
	
	\subsection{Delta consistency on regular fund regions}
	\label{subsec:delta-theory}
	
	On a compact regular fund region define
	$\Delta_n^h=\partial_FU_n^h$ and
	$\widetilde\Delta_n^h=\partial_F\widetilde U_n^h$. A regular region does not
	cross a preserved payoff kink or exercise boundary; at such singularities the
	statements apply to one-sided derivatives.
	
	Proposition~\ref{prop:delta-bound} controls the derivative error created by one
	cubic pruning operation. The following statements make explicit what this
	implies for a refinement sequence.
	
	\begin{corollary}
		\label{cor:fixed-delta-consistency}
		At a fixed pruning operation, suppose the source partition is fixed on a
		compact regular fund interval $K$ and has minimum cell length $m_K>0$. If
		$\|U-\widetilde U\|_{\infty,K}\le\varepsilon$, then
		\begin{equation}
			\|\partial_FU-\partial_F\widetilde U\|_{\infty,K}
			\le\frac{18\varepsilon}{m_K}.
			\label{eq:fixed-delta-consistency}
		\end{equation}
		Hence the pruned Delta converges uniformly to the source-curve Delta as
		$\varepsilon\downarrow0$.
	\end{corollary}
	
	\begin{proof}
		Apply Proposition~\ref{prop:delta-bound} on every cell of the common refinement
		and use $h_I\ge m_K$.
	\end{proof}
	
	\begin{theorem}
		\label{thm:unpruned-binomial-delta}
		Let $K$ be a compact interval of initial fund values. Assume the hypotheses of
		Theorem~\ref{thm:binomial-weak-order} and, in addition, that the exact backward
		solution and the strict binomial backward values are six-times continuously
		differentiable on an open tube containing every branch image generated from
		$K$. At each contractual date, assume that this tube stays a positive distance
		from payoff and exercise singularities, so the active branch of every event map
		is fixed there. If the one-step weak defect in
		\eqref{eq:binomial-local-expansion} and the operator-stability estimate hold in
		the $C^1$ norm on that tube, then
		\begin{equation}
			\|\partial_FU_0^{h,\mathrm{bin}}-\partial_FU_0\|_{\infty,K}
			\le C_K h.
			\label{eq:unpruned-binomial-delta-rate}
		\end{equation}
	\end{theorem}
	
	\noindent\emph{Proof.} See Appendix~\ref{app:proofs}.
	
	\begin{theorem}
		\label{thm:end-to-end-delta}
		Let $K$ be a compact fund interval at the initial factor state that stays a
		positive distance from every preserved payoff, guarantee, or exercise
		singularity. Let $U_0^h$ and $\widetilde U_0^h$ be the unpruned and adaptive
		initial-node curves. Assume
		\begin{enumerate}[label=(D\arabic*)]
			\item $U_0^h\to U_0$ in $C^1(K)$; for the strict binomial chain this follows
			from Theorem~\ref{thm:unpruned-binomial-delta} under its regular-tube
			hypotheses;
			\item on every cell $I$ of the common refinement of $U_0^h$ and
			$\widetilde U_0^h$, with length $h_I$, the final representation error
			$E_h(I)=\|U_0^h-\widetilde U_0^h\|_{\infty,I}$ satisfies
			\begin{equation}
				\max_{I\subset K}\frac{E_h(I)}{h_I}\longrightarrow0.
				\label{eq:delta-refinement-condition}
			\end{equation}
		\end{enumerate}
		Then
		\begin{equation}
			\|\partial_F\widetilde U_0^h-\partial_FU_0\|_{\infty,K}
			\longrightarrow0.
			\label{eq:end-to-end-delta}
		\end{equation}
		At a preserved singularity, the same conclusion applies separately to the left
		and right derivatives on adjacent regular intervals.
	\end{theorem}
	
	\begin{proof}
		By the triangle inequality,
		\begin{align*}
			\|\partial_F\widetilde U_0^h-\partial_FU_0\|_{\infty,K}
			\le{}&
			\|\partial_F\widetilde U_0^h-\partial_FU_0^h\|_{\infty,K}
			+\|\partial_FU_0^h-\partial_FU_0\|_{\infty,K}.
		\end{align*}
		The second term tends to zero by (D1). On each common-refinement cell,
		Proposition~\ref{prop:delta-bound} bounds the first term by
		$18E_h(I)/h_I$; taking the maximum and using (D2) proves
		\eqref{eq:end-to-end-delta}.
	\end{proof}
	
	\begin{remark}
		Condition (D1) is not implied by value convergence alone. Theorem~\ref{thm:unpruned-binomial-delta} supplies it for the strict binomial chain on regular tubes; for the production adaptive-trinomial sequence it remains an explicit regularity condition rather than a claimed rate theorem.
		Condition (D2) is a joint time--space requirement: a fixed value tolerance is
		not sufficient if the smallest relevant fund cell also shrinks. The stability
		bound \eqref{eq:global-representation-bound} supplies a computable sufficient
		upper bound for the numerator in (D2). The numerical tests in
		Section~\ref{subsec:delta-validation} assess the complete recursion rather than
		only one pruning operation.
	\end{remark}
	
	\subsection{Richardson extrapolation for the binomial scheme}
	\label{subsec:richardson-theory}
	
	Weak order one is not enough to justify Richardson extrapolation; a stable
	asymptotic error expansion is required. We state the assumptions explicitly
	because they delimit the theoretical claim.
	
	Let $Q_{t,h}$ be one discounted step of the strict, unpruned binomial chain and
	$P_{t,t+h}$ the exact discounted semigroup. Assume:
	\begin{enumerate}[label=(R\arabic*)]
		\item the correlation matrix is uniformly positive definite and the localized
		binomial joint probabilities are strictly admissible for all sufficiently
		small $h$;
		\item $\varphi$ is three-times continuously differentiable and bounded below, and
		the backward value and its images under the deterministic event maps have state
		derivatives up to order eight and time derivatives up to order three, all of
		polynomial growth;
		\item the numerical chain and diffusion have uniform polynomial moment bounds,
		and localization errors are $o(h^2)$ in the initial discounted value;
		\item for every test function in this regularity class,
		\begin{align}
			Q_{t,h}f&=f+h\mathcal A_tf+h^2\mathcal B_tf
			+h^3\mathcal C_tf+O(h^4),\label{eq:rich-num-expansion}\\
			P_{t,t+h}f&=f+h\mathcal A_tf+h^2\mathcal E_tf
			+h^3\mathcal F_tf+O(h^4),\label{eq:rich-exact-expansion}
		\end{align}
		with polynomially weighted, uniform remainders;
		\item the finite set of contractual dates is common to every refined grid and
		the associated interface maps preserve the stated regularity.
	\end{enumerate}
	For the binomial construction, the integer-power moment structure required by
	(R4) is obtained on interior compact sets by expanding $p_{Z,h}$ around $1/2$,
	inserting the explicit $2\times2$ coupling, and Taylor expanding the one-step
	operator. Odd standardized moments are
	$O(\!\sqrt h)$ and even moments have expansions in integer powers of $h$;
	therefore raw third and fifth moments first enter at orders $h^2$ and $h^3$.
	This is the discrete structure required by the Talay--Tubaro expansion
	\cite{TalayTubaro1990}.
	
	\begin{lemma}
		\label{lem:richardson-defect}
		Under (R2)--(R4), there are differential operators $\Psi_t$ and $\Xi_t$ such
		that
		\begin{equation}
			(Q_{t,h}-P_{t,t+h})f
			=h^2\Psi_tf+h^3\Xi_tf+O(h^4)
			\label{eq:richardson-defect}
		\end{equation}
		uniformly on the localized state domain.
	\end{lemma}
	
	\begin{proof}
		Subtract \eqref{eq:rich-exact-expansion} from
		\eqref{eq:rich-num-expansion} and set
		$\Psi_t=\mathcal B_t-\mathcal E_t$ and
		$\Xi_t=\mathcal C_t-\mathcal F_t$.
	\end{proof}
	
	\begin{theorem}
		\label{thm:richardson-expansion}
		Under (R1)--(R5), the strict, unpruned binomial price of a sufficiently smooth
		contract satisfies
		\begin{equation}
			\mathcal V_h=\mathcal V+c_1h+c_2h^2+o(h^2),
			\label{eq:global-richardson-expansion}
		\end{equation}
		where $c_1$ and $c_2$ are independent of the refinement level.
	\end{theorem}
	
	\noindent\emph{Proof.} See Appendix~\ref{app:proofs}.
	
	\begin{corollary}
		\label{cor:richardson}
		Under Theorem~\ref{thm:richardson-expansion},
		\begin{equation}
			\mathcal V_h^{\rm Rich}=2\mathcal V_{h/2}-\mathcal V_h=\mathcal V+O(h^2).
			\label{eq:richardson-standard}
		\end{equation}
		More generally, for a refinement ratio $r>1$,
		\begin{equation}
			\mathcal V_{h,r}^{\rm Rich}=\frac{r\mathcal V_{h/r}-\mathcal V_h}{r-1}=\mathcal V+O(h^2).
			\label{eq:richardson-general}
		\end{equation}
	\end{corollary}
	
	\begin{proof}
		Insert \eqref{eq:global-richardson-expansion} into either linear combination;
		the coefficient of $c_1h$ cancels identically.
	\end{proof}
	
	If the net-value function has a uniform expansion
	$\Phi_h(P)=\Phi(P)+a_1(P)h+a_2(P)h^2+o(h^2)$ near a simple fair-premium root,
	the implicit-function theorem gives
	$P_J^h=P_J+d_1h+d_2h^2+o(h^2)$, so the same extrapolation applies to
	the premium. For $P_{\mathrm{NS}}$, this follows directly from
	\eqref{eq:affine-fair-premium}, because $A^M$ is grid independent.
	
	\begin{remark}
		\label{rem:richardson-limits}
		Theorem~\ref{thm:richardson-expansion} is a theorem for strict, unpruned
		transitions and sufficiently smooth values. Projection and representation
		errors must be $o(h)$ to preserve the leading coefficient and $O(h^2)$ to
		retain the second-order extrapolation. The payoff $\max\{F,G\}$ and a Bermudan
		obstacle are non-smooth and can generate lattice-alignment oscillations
		\cite{Diener2004}. A full second-order theorem for the exact
		insurance payoff therefore requires a smoothing or interface argument in
		addition to (R1)--(R5). The binomial Richardson values reported below should
		be read as theoretically motivated and numerically validated accelerations,
		not as an unconditional second-order theorem for the kinked contract. The
		adaptive-trinomial sequence is not extrapolated because its non-nested triplet
		selection does not display a stable leading coefficient on the available
		grids.
	\end{remark}
	
	\section{Numerical study}\label{sec:numerics}
	
	\subsection{Setup and independent benchmarks}
	\label{subsec:setup}
	We use SP, MC, and LSMC for singular points, Monte Carlo, and least-squares
	Monte Carlo. In the diagnostic tables, $N^{\rm knot}_{\max}$ is the largest number of
	retained fund knots at a joint factor node, $\varepsilon_{\max}$ is the largest
	posterior continuous pruning error, and $t_{12}$ is twelve-thread wall time.
	Table~\ref{tab:baseline} summarizes the central specification. Unless otherwise
	stated, contract-level results use the production configuration: strict adaptive
	trinomial factor lattices, 24 steps per year ($N=480$), cubic fund functions,
	and $(\varepsilon_{\rm abs},\varepsilon_{\rm rel})=(10^{-4},10^{-6})$.
	Mortality uses the 2025 ISTAT resident-population table at issue age 50.
	
	\begin{table}[htbp]
		\centering
		\begin{singlespace}
			\small
			\renewcommand{\arraystretch}{1.15}
			\begin{tabular}{@{}lll@{}}
				\toprule
				Block & Parameters & Values \\
				\midrule
				Contract & $T,D,g_{m},g_{s},q$ & $20,100,0.02,0.02,0$ \\
				Heston variance & $V_0,\kappa_V,\theta_V,\sigma_V$ & $0.04,2.00,0.04,0.30$ \\
				CIR rate factor & $X_0,\kappa_r,\theta_r,\sigma_r$ & $0.018,0.45,0.025,0.10$ \\
				Correlations & $\rho_{SV},\rho_{Sr},\rho_{Vr}$ & $-0.70,-0.20,0.02$ \\
				Factor grid & type; $(\lambda_V,\lambda_X)$ & adaptive trinomial; $(1.0,1.4)$ \\
				Fund representation & anchors; tolerances & cubic; $10^{-4},10^{-6}$ \\
				Production grid & $N_{\rm yr},N$ & $24,480$ \\
				Surrender & dates; $\alpha_i^{\mathrm s}$ & annual; $1$ \\
				\bottomrule
			\end{tabular}
		\end{singlespace}
		\caption{Central contract, model, and numerical parameters. The initial curve is
			the ECB euro-area AAA government-bond Svensson curve dated 15 July 2026.}
		\label{tab:baseline}
	\end{table}
	
	The financial engine reproduces all grid-maturity discount factors and the
	corrected discounted-fund martingale condition to roundoff. The production
	trinomial transitions match the two CIR moments and all three correlations to
	roundoff and remain strictly admissible over the complete refinement range.
	Detailed factor-grid, admissibility, curve-fit, martingale, and refinement
	diagnostics are reported in the \supp.
	
	The independent no-surrender benchmark is antithetic full-truncation Euler Monte
	Carlo with exact CIR++ shift increments and analytic control variates. One
	hundred replications of one million paths at 256 time steps per year give
	$P_{\rm MC}=111.6586729$ with standard error $0.0012648$. All deterministic
	timings use 12 OpenMP threads on an Intel Core Ultra 5 125H notebook with 32 GB
	RAM; absolute times are machine-specific.
	
	For validation, we use independent Monte Carlo for the no-surrender benchmark
	and cross-fitted least-squares Monte Carlo for the surrender products. These
	approaches provide external reference values based on numerical procedures
	that are structurally different from the proposed deterministic method, and
	therefore offer a useful check on its accuracy without relying on the same
	discretization or approximation mechanisms.
	
	\subsection{Convergence and cost--accuracy}
	\label{subsec:convergence-numerics}
	Table~\ref{tab:temporal-convergence-main} reports representative levels of the
	common refinement sequence. It shows the slow monotone approach of the direct
	binomial values, the substantial cancellation obtained by Richardson, and the
	higher coarse-grid accuracy of the adaptive trinomial scheme.
	
	\begin{table}[htbp]
		\centering
		\begin{singlespace}
			\small
			\renewcommand{\arraystretch}{1.12}
			\begin{tabular}{rrrrr}
				\toprule
				$N_{\rm yr}$ & $N$ & Binomial direct & Binomial Richardson & Trinomial direct\\
				\midrule
				4  & 80  & 110.748278 & 111.157080 & 111.594769\\
				8  & 160 & 111.140506 & 111.532734 & 111.584134\\
				16 & 320 & 111.374914 & 111.609323 & 111.640138\\
				24 & 480 & 111.465744 & 111.636463 & 111.644323\\
				32 & 640 & 111.512365 & 111.649816 & 111.653842\\
				\bottomrule
			\end{tabular}
		\end{singlespace}
		\caption{Representative temporal refinements for the mortality-extended
			no-surrender premium. The independent Monte Carlo benchmark is
			$111.6586729$. Richardson uses the corresponding half-resolution binomial
			run.}
		\label{tab:temporal-convergence-main}
	\end{table}
	
	Figure~\ref{fig:cost-accuracy} compares the direct binomial, binomial Richardson,
	and direct adaptive-trinomial sequences. Direct binomial values approach the
	Monte Carlo benchmark from below. At $N=640$, Richardson reduces the absolute
	discrepancy to $0.008857$ in 116.2 seconds, while the direct trinomial value has
	error $0.004830$ in 191.3 seconds. The trinomial representation stays between
	176 and 187 knots at every displayed refinement. Thus Richardson is competitive
	at intermediate cost, whereas the adaptive trinomial method gives the most
	accurate direct value. Because the $\lambda_V=1$ minimal-span rule can select
	one-sided stencils, these direct trinomial refinements are used as a convergent
	accuracy sequence without fitting or claiming a formal weak order.
	
	\begin{figure}[htbp]
		\centering
		\includegraphics[width=\textwidth]{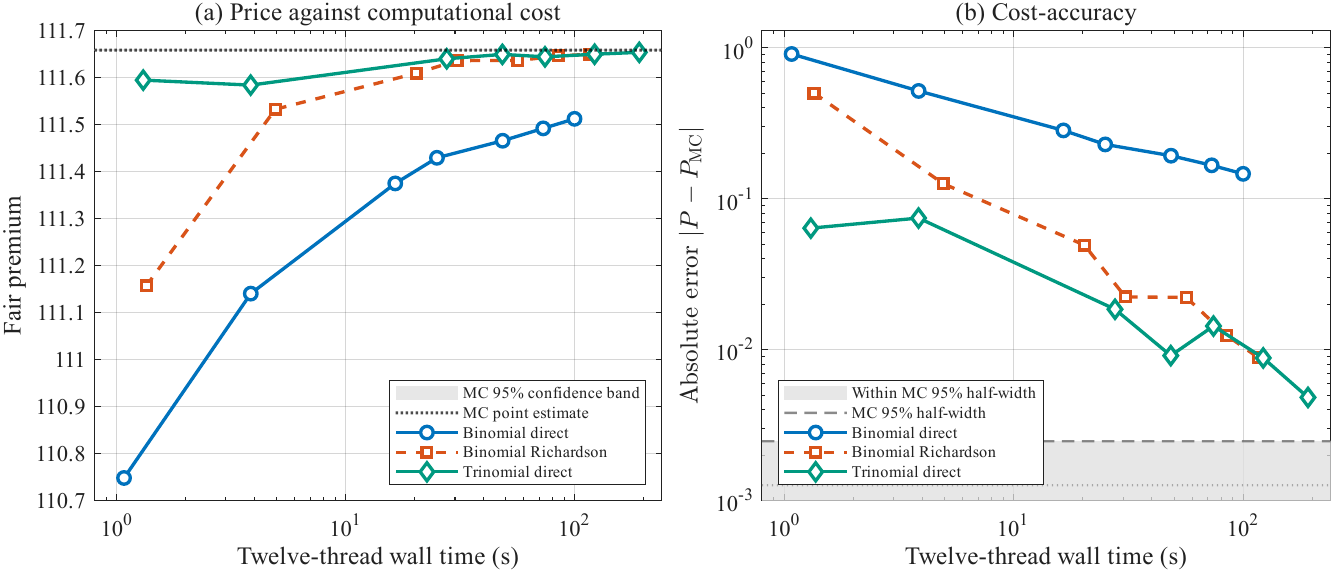}
		\caption{Cost--accuracy comparison for the mortality-extended no-surrender
			premium. Richardson cost includes both required grids; dotted step segments are
			best-so-far error envelopes.}
		\label{fig:cost-accuracy}
	\end{figure}
	\FloatBarrier
	
	\subsection{Pruning sensitivity and selected tolerance}
	\label{subsec:pruning-sensitivity}
	
	Table~\ref{tab:pruning-sensitivity-main} varies both coefficients of the mixed
	certificate at $N_{\rm yr}=24$ while keeping every other setting fixed. The
	comparison uses the exact interval-wise posterior certificate and disables all
	probability-dependent compression policies. Within each block, the tightest
	row is used only as an internal representation reference.
	
	\begin{table}[htbp]
		\centering
		\begin{singlespace}
			\footnotesize
			\renewcommand{\arraystretch}{1.10}
			\setlength{\tabcolsep}{3.5pt}
			\begin{minipage}[t]{0.48\textwidth}
				\centering
				\textit{Panel A: binomial scheme}\\[0.25em]
				\begin{tabular}{rccc}
					\toprule
					$\varepsilon$ & $P$ & $|P-P_{\rm tight}|$ & $N^{\rm knot}_{\max}$\\
					\midrule
					\multicolumn{4}{l}{$\varepsilon_{\rm abs}$; $\varepsilon_{\rm rel}=10^{-6}$}\\
					$3\cdot10^{-4}$ &111.465926&$2.155\cdot10^{-4}$&206\\
					$10^{-4}$       &111.465754&$4.395\cdot10^{-5}$&215\\
					$3\cdot10^{-5}$ &111.465728&$1.786\cdot10^{-5}$&218\\
					$10^{-5}$       &111.465710&0&220\\[0.2em]
					\multicolumn{4}{l}{$\varepsilon_{\rm rel}$; $\varepsilon_{\rm abs}=10^{-4}$}\\
					$3\cdot10^{-6}$ &111.467447&$2.215\cdot10^{-3}$&132\\
					$10^{-6}$       &111.465754&$5.222\cdot10^{-4}$&215\\
					$3\cdot10^{-7}$ &111.465345&$1.126\cdot10^{-4}$&373\\
					$10^{-7}$       &111.465232&0&616\\
					\bottomrule
				\end{tabular}
			\end{minipage}\hfill
			\begin{minipage}[t]{0.48\textwidth}
				\centering
				\textit{Panel B: trinomial scheme}\\[0.25em]
				\begin{tabular}{rccc}
					\toprule
					$\varepsilon$ & $P$ & $|P-P_{\rm tight}|$ & $N^{\rm knot}_{\max}$\\
					\midrule
					\multicolumn{4}{l}{$\varepsilon_{\rm abs}$; $\varepsilon_{\rm rel}=10^{-6}$}\\
					$3\cdot10^{-4}$ &111.644490&$2.296\cdot10^{-4}$&181\\
					$10^{-4}$       &111.644323&$6.213\cdot10^{-5}$&187\\
					$3\cdot10^{-5}$ &111.644267&$6.592\cdot10^{-6}$&188\\
					$10^{-5}$       &111.644261&0&190\\[0.2em]
					\multicolumn{4}{l}{$\varepsilon_{\rm rel}$; $\varepsilon_{\rm abs}=10^{-4}$}\\
					$3\cdot10^{-6}$ &111.645861&$2.063\cdot10^{-3}$&121\\
					$10^{-6}$       &111.644323&$5.249\cdot10^{-4}$&187\\
					$3\cdot10^{-7}$ &111.643913&$1.152\cdot10^{-4}$&320\\
					$10^{-7}$       &111.643798&0&489\\
					\bottomrule
				\end{tabular}
			\end{minipage}
		\end{singlespace}
		\caption{Pruning sensitivity at $N_{\rm yr}=24$ ($N=480$). The upper block
			in each panel varies $\varepsilon_{\rm abs}$ and the lower block varies
			$\varepsilon_{\rm rel}$. $P_{\rm tight}$ is the tightest result within the
			corresponding block.}
		\label{tab:pruning-sensitivity-main}
	\end{table}
	
	Tightening $\varepsilon_{\rm abs}$ from $10^{-4}$ to $10^{-5}$ changes the
	premium by $4.40\times10^{-5}$ in the binomial scheme and
	$6.21\times10^{-5}$ in the trinomial scheme. The relative coefficient has a
	larger effect on complexity because its contribution grows with $|F|$.
	Relative to $10^{-7}$, the production value $\varepsilon_{\rm rel}=10^{-6}$
	changes the premium by about $5.2\times10^{-4}$ in either scheme while keeping
	$N^{\rm knot}_{\max}$ at 215 and 187; the $10^{-7}$ calculations require 616 and 489
	knots. A separate fund-cap check over $10^8D$--$10^{12}D$ has spread
	$4.21\times10^{-6}$. These comparisons motivate the production setting
	$(\varepsilon_{\rm abs},\varepsilon_{\rm rel})=(10^{-4},10^{-6})$.
	
	\subsection{Delta consistency and independent no-surrender benchmark}
	\label{subsec:delta-validation}
	
	Table~\ref{tab:delta-full-main} combines three tests. Panel A isolates one
	local pruning operation on the conditional policy-year-10 value. Panel B adds
	temporal refinement and compares the analytic cubic Delta with an independent
	pathwise Monte Carlo benchmark on
	$\mathcal F=\{500,1000,1500,2000\}$. These four fund values span a lower-fund
	region, the representative policy-year-10 state used in the hedging experiment,
	and progressively more fund-dominated regions. Panel C reports the inception
	Delta immediately after the first contribution. For Panel A,
	$E_U=\max_{F\in\mathcal F}|U^\varepsilon(F)-U^0(F)|$ and
	$E_\Delta=\max_{F\in\mathcal F}|\Delta^\varepsilon(F)-\Delta^0(F)|$.
	The Monte Carlo benchmarks use 1.6 million antithetic paths and 256
	full-truncation Euler steps per year; for Panel C,
	$\Delta_{0,\rm MC}=0.810044207$ with standard error $1.70\times10^{-4}$.
	
	\begin{table}[htbp]
		\centering
		\begin{singlespace}
			\footnotesize
			\renewcommand{\arraystretch}{1.10}
			\setlength{\tabcolsep}{4.2pt}
			\begin{minipage}[t]{0.43\textwidth}
				\centering
				\textit{Panel A: local pruning at policy year 10}\\[0.25em]
				\begin{tabular}{rccc}
					\toprule
					$\varepsilon_{\rm abs}$ & $E_U$ & $E_\Delta$ & knots\\
					\midrule
					$10^{-3}$ & $1.295\cdot10^{-4}$ & $1.508\cdot10^{-5}$ & 18\\
					$10^{-4}$ & $3.536\cdot10^{-5}$ & $1.162\cdot10^{-6}$ & 26\\
					$10^{-5}$ & $2.780\cdot10^{-6}$ & $1.675\cdot10^{-6}$ & 52\\
					\bottomrule
				\end{tabular}
			\end{minipage}\hfill
			\begin{minipage}[t]{0.53\textwidth}
				\centering
				\textit{Panel B: policy-year-10 temporal convergence}\\[0.25em]
				\begin{tabular}{rccc}
					\toprule
					$N_{\rm yr}$ & Binomial & Richardson & Trinomial\\
					\midrule
					4  & $1.682\cdot10^{-2}$ & $1.005\cdot10^{-2}$ & $2.076\cdot10^{-3}$\\
					8  & $9.367\cdot10^{-3}$ & $1.918\cdot10^{-3}$ & $1.165\cdot10^{-3}$\\
					16 & $4.879\cdot10^{-3}$ & $3.900\cdot10^{-4}$ & $5.674\cdot10^{-4}$\\
					24 & $3.135\cdot10^{-3}$ & $1.377\cdot10^{-4}$ & $4.788\cdot10^{-4}$\\
					32 & $2.230\cdot10^{-3}$ & $4.187\cdot10^{-4}$ & $6.682\cdot10^{-4}$\\
					\bottomrule
				\end{tabular}
			\end{minipage}
			
			\vspace{0.8em}
			\textit{Panel C: inception Delta after the first contribution}\\[0.25em]
			\begin{tabular}{rcccccc}
				\toprule
				&\multicolumn{2}{c}{Binomial}&\multicolumn{2}{c}{Richardson}&\multicolumn{2}{c}{Trinomial}\\
				\cmidrule(lr){2-3}\cmidrule(lr){4-5}\cmidrule(lr){6-7}
				$N_{\rm yr}$&$\Delta_0$&$e_0$&$\Delta_0$&$e_0$&$\Delta_0$&$e_0$\\
				\midrule
				4  &0.808856&$1.188\cdot10^{-3}$&0.807713&$2.332\cdot10^{-3}$&0.811659&$1.615\cdot10^{-3}$\\
				8  &0.809175&$8.688\cdot10^{-4}$&0.809495&$5.493\cdot10^{-4}$&0.810985&$9.412\cdot10^{-4}$\\
				16 &0.809589&$4.552\cdot10^{-4}$&0.810003&$4.166\cdot10^{-5}$&0.810557&$5.132\cdot10^{-4}$\\
				24 &0.809775&$2.693\cdot10^{-4}$&0.810120&$7.615\cdot10^{-5}$&0.810507&$4.630\cdot10^{-4}$\\
				32 &0.809888&$1.566\cdot10^{-4}$&0.810186&$1.420\cdot10^{-4}$&0.810460&$4.162\cdot10^{-4}$\\
				\bottomrule
			\end{tabular}
		\end{singlespace}
		\caption{Delta consistency of the certified cubic no-surrender engine.
			Panel A reports the local value and Delta errors relative to the unpruned
			source curve. Panel B reports
			$E_\Delta=\max_{F\in\mathcal F}|\Delta(F)-\Delta_{\rm MC}(F)|$.
			Panel C reports $e_0=|\Delta_0-\Delta_{0,\rm MC}|$; the 95\% Monte Carlo
			half-width at inception is $0.000333$.}
		\label{tab:delta-full-main}
	\end{table}
	
	Panel A confirms that a value-only certificate preserves the local derivative:
	all displayed Delta errors are below $1.6\times10^{-5}$. In Panel B, the
	direct binomial error decreases steadily and Richardson reduces the maximum
	discrepancy to $1.38\times10^{-4}$ at $N_{\rm yr}=24$. The mild
	non-monotonicity on the finest Richardson and trinomial levels is comparable to
	the Monte Carlo sampling uncertainty and does not include Monte Carlo time
	bias. Panel C gives the same conclusion at the initial hedging point:
	Richardson lies within one Monte Carlo standard error from
	$N_{\rm yr}=16$ onward, while the trinomial discrepancy decreases monotonically
	to $4.16\times10^{-4}$.
	
	\subsection{Correlation sensitivity and local admissibility}
	\label{subsec:correlation-sensitivity-main}
	
	The fully correlated specification is numerically useful only if the local
	coupling remains strict over economically relevant parameter changes.
	Figure~\ref{fig:correlation-sensitivity-main} varies one correlation at a time
	around the baseline
	$(\rho_{SV},\rho_{Sr},\rho_{Vr})=(-0.70,-0.20,0.02)$ under the production
	configuration. The ordinate is the premium change relative to
	$P_{\mathrm{NS}}^{\rm base}=111.644322965$.
	
	\begin{figure}[htbp]
		\centering
		\includegraphics[width=0.94\textwidth]{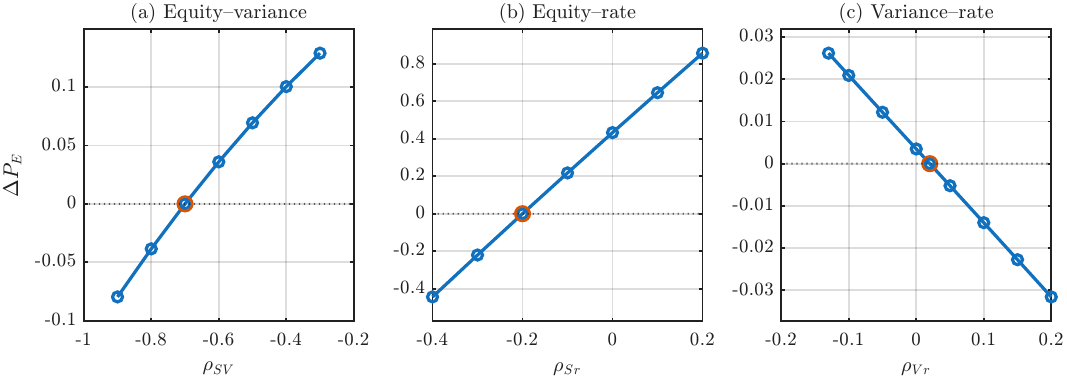}
		\caption{Sensitivity of the no-surrender fair premium to the three pairwise
			correlations. In each panel, the displayed correlation varies while the other
			two remain at their baseline values.}
		\label{fig:correlation-sensitivity-main}
	\end{figure}
	
	The dependence is close to linear over the displayed ranges. Least-squares
	slopes are approximately $0.348$, $2.166$, and $-0.175$ with respect to
	$\rho_{SV}$, $\rho_{Sr}$, and $\rho_{Vr}$, respectively. Thus a change of
	$0.1$ in the equity--rate correlation changes the annual premium by about
	$0.217$, compared with $0.035$ for the equity--variance correlation and
	$-0.017$ for the variance--rate correlation. The variance--rate sweep is
	restricted to $[-0.13,0.20]$, well inside the common strict range in
	Proposition~\ref{prop:trinomial-rho-range-main}; no point in the sweep uses
	local projection. Direct binomial and binomial Richardson sweeps give the same
	qualitative ordering.
	\FloatBarrier
	
	\subsection{Surrender contracts and independent LSMC validation}
	\label{subsec:surrender-results}
	For $J\in\{\mathrm{FB},\mathrm{GB},\mathrm{MX}\}$, define
	$H_J=P_J-P_{\mathrm{NS}}$ as the surrender-option component of the fair
	annual premium.
	Table~\ref{tab:surrender-compact} summarizes the contract-level comparison. The
	surrender-option components are considerably more stable than the premium
	levels: across the five finest trinomial grids they lie in
	$H_{\mathrm{FB}}\in[0.2201,0.2209]$, $H_{\mathrm{GB}}\in[5.7745,5.7792]$, and
	$H_{\mathrm{MX}}\in[6.4243,6.4287]$.
	
	\begin{table}[htbp]
		\centering
		\begin{singlespace}
			\footnotesize
			\renewcommand{\arraystretch}{1.13}
			\textit{Panel A: selected deterministic valuations}\\[0.3em]
			\begin{tabular}{lrrrr}
				\toprule
				Method & $P_{\mathrm{NS}}$ & $P_{\mathrm{FB}}$ & $P_{\mathrm{GB}}$ & $P_{\mathrm{MX}}$\\
				\midrule
				Binomial Richardson, $N_{\rm yr}=32$ & 111.6499 & 111.8719 & 117.4202 & 118.0735\\
				Trinomial direct, $N_{\rm yr}=24$ & 111.6443 & 111.8644 & 117.4207 & 118.0693\\
				Trinomial direct, $N_{\rm yr}=32$ & 111.6538 & 111.8744 & 117.4331 & 118.0826\\
				\midrule
				\multicolumn{5}{l}{\textit{Surrender-option components}}\\
				Binomial Richardson, $N_{\rm yr}=32$ & -- & 0.2220 & 5.7703 & 6.4236\\
				Trinomial direct, $N_{\rm yr}=24$ & -- & 0.2201 & 5.7763 & 6.4250\\
				Trinomial direct, $N_{\rm yr}=32$ & -- & 0.2206 & 5.7792 & 6.4287\\
				\bottomrule
			\end{tabular}
			
			\vspace{0.75em}
			\textit{Panel B: independent LSMC validation of surrender premiums}\\[0.3em]
			\begin{tabular*}{0.88\textwidth}{@{\extracolsep{\fill}}lrcrr@{}}
				\toprule
				Product & LSMC--256 & Held-out 95\% CI & SP production & SP--LSMC\\
				\midrule
				$P_{\mathrm{FB}}$ & 111.7090 & $[111.6279,111.7902]$ & 111.8644 & 0.1554\\
				$P_{\mathrm{GB}}$ & 117.3181 & $[117.1971,117.4390]$ & 117.4207 & 0.1026\\
				$P_{\mathrm{MX}}$ & 117.7707 & $[117.6766,117.8647]$ & 118.0693 & 0.2986\\
				\bottomrule
			\end{tabular*}
		\end{singlespace}
		\caption{Fair annual premiums and surrender-option components. Panel B is an
			independent two-fold cross-fitted validation using the guarantee-kink quadratic
			basis, one million paths, and 256 time steps per policy year. Its intervals
			quantify held-out sampling uncertainty only; regression and time-discretization
			bias are not included.}
		\label{tab:surrender-compact}
	\end{table}
	
	Every deterministic refinement gives $P_{\mathrm{MX}}>P_{\mathrm{GB}}>P_{\mathrm{FB}}>P_{\mathrm{NS}}$. The $N_{\rm yr}=32$
	binomial Richardson components closely match the direct trinomial values,
	providing an internal deterministic cross-check. The cross-fitted LSMC premiums
	are lower by $0.1554$, $0.1026$, and $0.2986$ for $P_{\mathrm{FB}},P_{\mathrm{GB}},P_{\mathrm{MX}}$, respectively;
	the largest relative discrepancy is $0.253\%$. This is consistent with a fitted,
	generally suboptimal stopping policy. In the associated simulation control, the
	no-surrender estimate is $111.6789$ with held-out interval
	$[111.5842,111.7736]$, containing both the production value and the
	high-precision Monte Carlo benchmark.
	
	Figure~\ref{fig:surrender-regions} shows policy-year-10 sections of the obstacle
	problem. Fund-based surrender occurs only in the high-fund tail,
	guarantee-based surrender in the low-fund region, and the mixed payoff on both
	sides of a central continuation band. Higher variance generally expands
	continuation, while higher rates favor exercise.
	
	\begin{figure}[htbp]
		\centering
		\includegraphics[width=\textwidth]{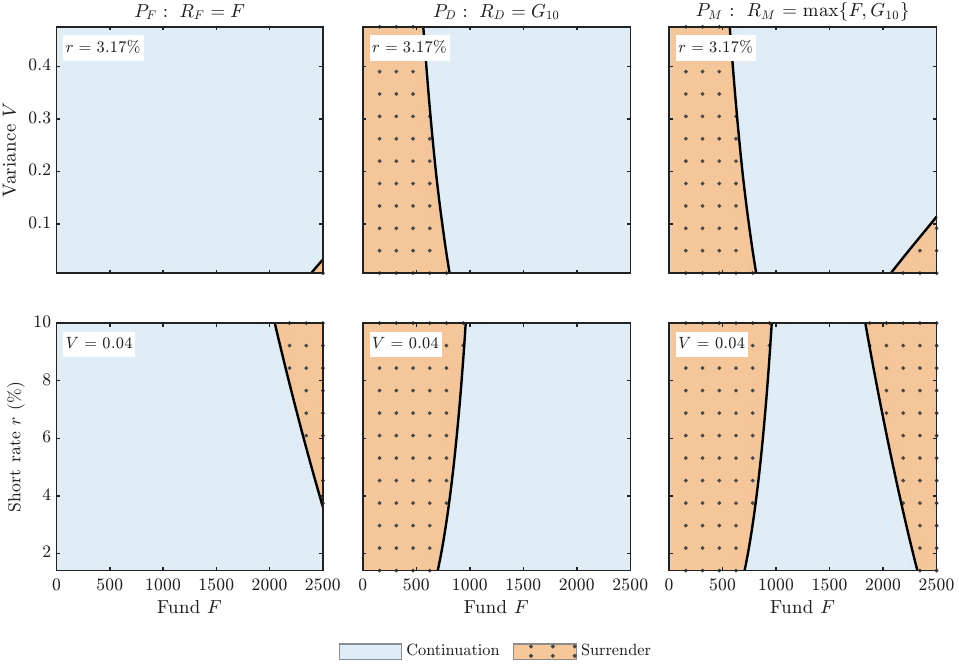}
		\caption{Continuation and surrender regions at policy year 10 under the
			production trinomial scheme. Pale blue denotes continuation, dotted orange
			immediate surrender, and black curves the interpolated boundaries.}
		\label{fig:surrender-regions}
	\end{figure}
	\FloatBarrier
	
	\subsection{Issue-age and mortality robustness}
	\label{subsec:mortality-robustness-main}
	
	The production configuration is also stable across materially different
	mortality levels. Table~\ref{tab:mortality-main} compares issue ages 40, 50,
	and 60 with a no-mortality benchmark while holding the financial model and
	numerical discretization fixed.
	
	\begin{table}[htbp]
		\centering
		\begin{singlespace}
			\small
			\renewcommand{\arraystretch}{1.12}
			\textit{Panel A: fair annual premiums}\\[0.25em]
			\begin{tabular}{lrrrrr}
				\toprule
				Scenario & 20-year survival & $P_{\mathrm{NS}}$ & $P_{\mathrm{FB}}$ & $P_{\mathrm{GB}}$ & $P_{\mathrm{MX}}$\\
				\midrule
				Issue age 40 & 96.0843\% &111.6528&111.8774&117.5034&118.1701\\
				Issue age 50 & 90.4215\% &111.6443&111.8644&117.4207&118.0693\\
				Issue age 60 & 75.6815\% &111.6206&111.8287&117.1937&117.7952\\
				No mortality &100.0000\% &111.6588&111.8864&117.5613&118.2410\\
				\bottomrule
			\end{tabular}
			
			\vspace{0.65em}
			\textit{Panel B: surrender-option components}\\[0.25em]
			\begin{tabular}{lrrr}
				\toprule
				Scenario & $H_{\mathrm{FB}}$ & $H_{\mathrm{GB}}$ & $H_{\mathrm{MX}}$\\
				\midrule
				Issue age 40 &0.2246&5.8506&6.5173\\
				Issue age 50 &0.2201&5.7763&6.4250\\
				Issue age 60 &0.2081&5.5730&6.1745\\
				No mortality &0.2276&5.9025&6.5822\\
				\bottomrule
			\end{tabular}
		\end{singlespace}
		\caption{Issue-age and mortality sensitivity under the production
			adaptive-trinomial discretization. The no-mortality row disables both mortality
			termination and the death benefit.}
		\label{tab:mortality-main}
	\end{table}
	
	The base no-surrender premium is only mildly affected, whereas mortality has a
	more visible effect on the guarantee-based surrender components. From the
	no-mortality case to issue age 60, $H_{\mathrm{GB}}$ and $H_{\mathrm{MX}}$ decrease by about 5.6\% and
	6.2\%, respectively, because death shortens the horizon over which the
	surrender guarantees can be exercised. The monotone age pattern confirms that
	the central results are not specific to the issue-age-50 calibration.
	\FloatBarrier
	
	\subsection{Risk-neutral timing of optimal surrender}
	\label{subsec:surrender-timing-main}
	
	The adaptive representation also gives direct access to the deterministic
	exercise policy. To summarize its dynamic implications, we simulate the
	production policies under the risk-neutral measure with two million paths and
	annual exercise only. Figure~\ref{fig:surrender-timing-main} reports the
	unconditional probability of first surrender in each policy year; the paths
	that do not surrender terminate through death or maturity.
	
	\begin{table}[htbp]
		\centering
		\begin{singlespace}
			\small
			\renewcommand{\arraystretch}{1.12}
			\begin{tabular}{lrrrr}
				\toprule
				Contract & Surrender & Death & Maturity & Mean surrender year\\
				\midrule
				$P_{\mathrm{FB}}$ & 29.9308\% & 8.5539\% & 61.5153\% & 15.825\\
				$P_{\mathrm{GB}}$ & 66.7444\% & 4.7233\% & 28.5322\% & 7.871\\
				$P_{\mathrm{MX}}$ & 89.7464\% & 3.5006\% & 6.7530\% & 9.001\\
				\bottomrule
			\end{tabular}
		\end{singlespace}
		\caption{Risk-neutral terminal outcomes under the fixed production surrender
			policies. The largest annual 95\% Monte Carlo half-width is 0.0392 percentage
			points.}
		\label{tab:surrender-timing-summary-main}
	\end{table}
	
	\begin{figure}[htbp]
		\centering
		\includegraphics[width=0.92\textwidth]{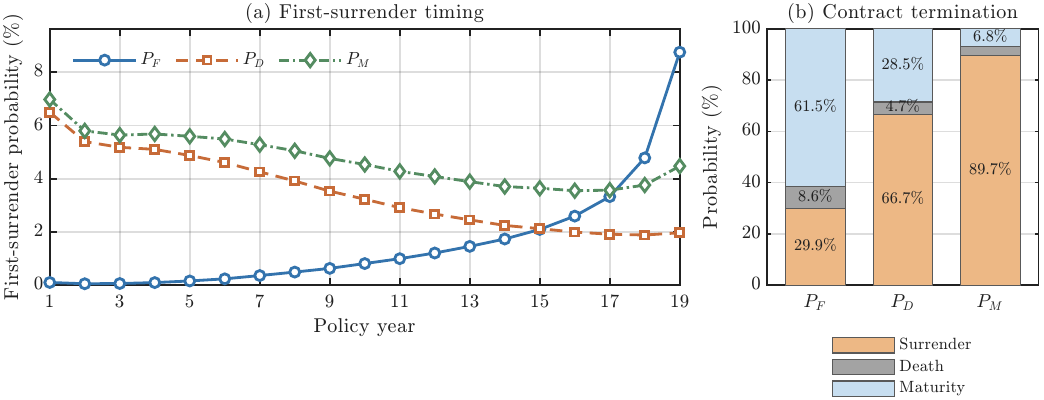}
		\caption{Unconditional risk-neutral distribution of the first surrender year
			under the production policies. The remaining probability mass terminates
			through death or maturity as reported in
			Table~\ref{tab:surrender-timing-summary-main}.}
		\label{fig:surrender-timing-main}
	\end{figure}
	
	Fund-based surrender is comparatively infrequent and late: more than 61\% of
	$P_{\mathrm{FB}}$ paths reach maturity, and conditional surrender occurs on average near
	policy year 16. Guarantee-based surrender is much earlier and more frequent.
	The mixed payoff is exercised on almost 90\% of paths, but its mean surrender
	year is later than for $P_{\mathrm{GB}}$ because the fund branch creates an additional
	continuation region at intermediate fund levels. No simulated observation
	falls outside the stored policy domain. These results are descriptive
	risk-neutral policy diagnostics, not behavioral lapse predictions.
	\FloatBarrier
	
	\subsection{Model relevance and factor hedging}
	\label{subsec:model-relevance}
	In this subsection, $L$ is the liability value, $C$ the one-year call value,
	and $B$ the zero-coupon bond value. Subscripts denote partial derivatives with
	respect to the indicated state variable, and $n_C,n_B,n_S$ are the call, bond,
	and stock holdings. For strategy $h$, $\Pi_h(t)$ is the hedge-portfolio value
	and $\mathrm{HE}_h(t)=\Pi_h(t)-L(t)$ is its mark-to-model hedging error.
	Table~\ref{tab:model-decomposition} compares four nested models while holding the
	contract, mortality, initial discount curve, and numerical engine fixed.
	Deterministic-rate models follow the same market forward curve, so the comparison
	is not contaminated by a different initial term structure.
	
	\begin{table}[htbp]
		\centering
		\begin{singlespace}
			\small
			\renewcommand{\arraystretch}{1.13}
			\begin{tabular}{lrrrr}
				\toprule
				Model & $P_{\mathrm{NS}}$ & $P_{\mathrm{FB}}$ & $P_{\mathrm{GB}}$ & $P_{\mathrm{MX}}$\\
				\midrule
				Black--Scholes, deterministic rates & 112.1970 & 112.6570 & 116.8521 & 117.8997\\
				Heston, deterministic rates & 111.8767 & 112.1013 & 117.3195 & 117.9572\\
				Black--Scholes--CIR++ & 111.9231 & 112.3773 & 117.0078 & 118.0790\\
				Heston--CIR++ & 111.6443 & 111.8644 & 117.4207 & 118.0693\\
				\bottomrule
			\end{tabular}
		\end{singlespace}
		\caption{Model-component comparison under the production numerical
			configuration and a common initial discount curve.}
		\label{tab:model-decomposition}
	\end{table}
	
	Relative to Black--Scholes with deterministic rates, the full model changes
	$(P_{\mathrm{NS}},P_{\mathrm{FB}},P_{\mathrm{GB}},P_{\mathrm{MX}})$ by approximately $(-0.553,-0.793,+0.569,+0.170)$. Hence
	model enrichment is payoff dependent rather than a uniform price correction.
	For the mixed contract, the volatility and rate contributions partly offset.
	
	To assess the risk-management relevance, we form a one-year static hedge for the
	no-surrender liability at policy year 10. Physical-measure paths are simulated
	under a stated drift scenario, while all liability values and hedge ratios are
	computed under $\mathbb Q$. We compare a Black--Scholes Delta hedge, a
	Heston--CIR++ Delta hedge, and a stock--call--bond hedge matching first-order
	fund, variance, and rate exposures. The positions are held fixed to policy year
	11; mortality is treated as diversified.
	
	\begin{table}[htbp]
		\centering
		\begin{singlespace}
			\small
			\begin{tabular}{lrr}
				\toprule
				Strategy & RMSE & Left-tail mean (1\%)\\
				\midrule
				Black--Scholes Delta & 36.549 & -92.222\\
				Heston--CIR++ Delta & 35.825 & -95.107\\
				Delta + variance + rate & 7.019 & -23.687\\
				\bottomrule
			\end{tabular}
		\end{singlespace}
		\caption{One-period mark-to-model hedging errors. The error is hedge value minus
			gross liability; the last column is the conditional mean below the empirical
			1\% quantile.}
		\label{tab:hedging}
	\end{table}
	
	Replacing the Black--Scholes Delta alone reduces RMSE by about 2\%. Matching
	variance and rate exposures reduces RMSE by 80.4\% relative to the
	model-consistent Delta-only hedge and reduces the magnitude of the left-tail
	mean by 75.1\%. This is a static, frictionless factor-exposure diagnostic, not a
	complete asset--liability-management study.
	
	\section{Conclusions}\label{sec:conclusions}
	
	We proposed an adaptive dynamic-programming method for long-dated
	annual-premium contracts with stochastic volatility, stochastic interest rates,
	mortality, and Bermudan surrender. The central state-space reduction is to
	discretize only the recombining financial factors and retain the accumulated
	fund as an adaptive value function. Contributions become translations, and the
	exercise decision is applied directly to the same function. This avoids a
	global fund grid while preserving a fully deterministic backward recursion.
	
	The numerical construction combines locally covariance-consistent square-root
	lattices, residual equity quadrature, discrete curve fitting, and a martingale
	correction. Piecewise-cubic singular-point propagation retains payoff and
	exercise kinks, exact affine exercise compression removes redundant regions,
	and continuous pruning certificates control the represented value between
	knots. Both strict financial chains converge weakly. The strict binomial chain is
	weakly first-order accurate and admits the asymptotic expansion required for
	Richardson under the stated smoothness assumptions; the production adaptive
	trinomial is retained as a direct convergent approximation without a claimed
	global rate. The adaptive valuation converges when representation and boundary
	errors vanish, and the cubic certificate yields a regular-region Delta
	consistency result under an explicit value-error-to-cell-length condition.
	
	The experiments support the method from complementary directions. The
	adaptive-trinomial scheme gives the most accurate direct no-surrender value on
	the finest grid, with fewer than 200 retained fund knots at every displayed
	node. Binomial Richardson provides a useful independent deterministic
	cross-check. Pathwise Monte Carlo supports the reported Deltas, while
	cross-fitted LSMC validates the ordering and level of the surrender premiums.
	The model-decomposition and hedge experiments show why small premium changes
	should not be interpreted as negligible model risk: the sign of the price effect
	depends on the contractual payoff, and unhedged variance and rate exposures can
	dominate a Delta-only hedge.
	
	The method is designed for a finite set of contractual and exercise dates. Two
	extensions are particularly natural: reusing adaptive representations across
	neighboring premium evaluations, and extending the error expansion to a broader
	class of non-smooth stopping problems. The same functional-state principle may
	also be useful in other stochastic dynamic programs in which deterministic cash
	additions prevent one state variable from recombining.
	
	\section*{Declarations}
	\begingroup
	\small
	\noindent\textbf{Funding.}
	This research received no specific grant from funding agencies in the public,
	commercial, or not-for-profit sectors.
	
	\smallskip
	\noindent\textbf{Competing interests.}
	The authors declare that they have no competing interests.
	
	\smallskip
	\noindent\textbf{Data availability.}
	The numerical inputs and aggregate results supporting the findings are reported
	in the article and the Supplementary Material.
	
	\smallskip
	\noindent\textbf{AI-assisted editing.}
	During the preparation of this work, the authors used ChatGPT 5.6 Sol (OpenAI) solely
	to assist with language, clarity, and textual organization. The authors
	reviewed and edited the resulting text and take full responsibility for the
	content of the publication.
	\endgroup
	
	\begin{singlespace}
		
	\end{singlespace}
	
	\appendix
	\section{Proofs of technical results}
	\label{app:proofs}
	
	This appendix collects the technical proofs deferred from the construction and
	convergence sections. The shorter arguments that are useful for following the
	main logical steps remain in the body of the paper.
	
	\subsection{Adaptive-trinomial factor construction}
	
	\begin{proof}[Proof of Lemma~\ref{lem:trinomial-interior-probabilities-main}]
		The exact CIR mean and variance satisfy, uniformly on $K_Z$,
		\begin{align*}
			m_h(z)&=z+\kappa(\theta-z)h+O(h^2),\\
			s_h^2(z)&=\sigma^2zh+O(h^2),
		\end{align*}
		while adjacent grid values obey
		$Z_{q\pm1}-Z_q=\pm\lambda\sigma\sqrt{zh}
		+\lambda^2\sigma^2h/4$.
		Write the three probabilities as the solution of the linear system matching
		unit mass, $m_h(z)$, and $s_h^2(z)+m_h(z)^2$. Substitution and expansion of the
		explicit solution give \eqref{eq:tri-prob-lambda-gt-one}--
		\eqref{eq:tri-lambda-one-side}. For $\lambda>1$, every leading probability is
		strictly positive. For $\lambda=1$, the leading central mass vanishes and its
		first non-zero coefficient is $A(z)$; if it is negative, the symmetric stencil
		is inadmissible. On the standardized limiting supports $(-2,0,1)$ and
		$(-1,0,2)$, matching zero mean and unit variance gives respectively
		$(1/6,1/2,1/3)$ and its reversal. Strict positivity of these limits makes the
		corresponding exact CIR stencils admissible for sufficiently small $h$.
		
	\end{proof}
	
	\begin{proof}[Proof of Proposition~\ref{prop:trinomial-rho-range-main}]
		For the rate factor, Lemma~\ref{lem:trinomial-interior-probabilities-main}
		implies that the standardized innovations converge uniformly to
		$(-1.4,0,1.4)$. For the variance factor, they converge either to
		$(-1,0,1)$ or, up to reflection, to $(-2,0,1)$. Hence, for every
		$\epsilon>0$ and all sufficiently small $h$,
		\begin{equation}
			\max_{i,j}|z_i^Vz_j^X|\le(2+\epsilon)(1.4+\epsilon)
		\end{equation}
		uniformly on $K$. If $|\rho_{Vr}|<5/14$, choose $\epsilon$ so that
		$|\rho_{Vr}|(2+\epsilon)(1.4+\epsilon)<1$. Then
		$1+\rho_{Vr}z_i^Vz_j^X>0$ for every active pair, which is precisely
		\eqref{eq:lancaster-nonnegative-main}.
		
	\end{proof}
	
	\subsection{Projected binomial chain}
	
	\begin{proof}[Proof of Proposition~\ref{prop:uniform-interior-admissibility}]
		The functions $\beta_V$ and $\beta_X$ in
		\eqref{eq:binomial-prob-expansion} are bounded on $K$. Expanding
		\eqref{eq:binomial-correlation-endpoints} around
		$p_V=p_X=1/2$ therefore gives, uniformly on $K$,
		\begin{align*}
			\rho_-(v,x;h)
			&=-1+2|\beta_V(v)+\beta_X(x)|\sqrt h+O(h),\\
			\rho_+(v,x;h)
			&=1-2|\beta_V(v)-\beta_X(x)|\sqrt h+O(h).
		\end{align*}
		Every fixed interior correlation lies between these endpoints for sufficiently
		small $h$.
		
	\end{proof}
	
	\begin{proof}[Proof of Proposition~\ref{prop:vanishing-projection}]
		Fix $\eta>0$. Under the strict Feller conditions, both continuous CIR factors
		remain positive. Compact containment of the diffusions and tightness of the
		marginal lattices allow a rectangle
		$K=[\underline v,\overline v]\times[\underline x,\overline x]$ that is a compact subset of
		$(0,\infty)\times(0,\infty)$ such that the probability that either the diffusion or the
		marginal chains leave $K$ before $T$ is below $\eta$ for all sufficiently
		small $h$. By Proposition~\ref{prop:uniform-interior-admissibility},
		$A_h\cap K=\varnothing$ for sufficiently small $h$. A projected node can then
		be reached only after the joint chain leaves $K$, so
		$\limsup_{h\downarrow0}\mathbb P(\tau_h\le T)\le\eta$. Letting
		$\eta\downarrow0$ proves the claim.
		
	\end{proof}
	
	\begin{proof}[Proof of Theorem~\ref{thm:projected-chain}]
		On each interior compact set, projected and strict local generators coincide
		for all sufficiently small $h$ except for the mixed $V$--$X$ covariance on
		$A_h$. For a smooth compactly supported test function, the integrated
		generator discrepancy is bounded by a constant times $\mathcal M_h$, which
		vanishes. Compact containment controls the complement of the localization
		set. The martingale-problem characterization identifies every subsequential
		limit with the Heston--CIR++ diffusion; uniqueness gives convergence of the
		whole sequence. Uniform integrability permits passage to discounted
		expectations. A finite collection of contractual and exercise dates is then
		handled by backward induction through the stable event maps.
		
	\end{proof}
	
	\subsection{Weak-error rates}
	
	\begin{proof}[Proof of Lemma~\ref{lem:binomial-scaled-moments}]
		On an interior compact set, the two binomial CIR successors are
		$O(\!\sqrt h)$ from the parent and the exact-mean interpolation probability
		has the expansion $p_{Z,h}=1/2+O(\!\sqrt h)$. The resulting conditional
		variance equals the CIR variance to leading order, with an $O(h^2)$ remainder.
		The centered variance-one innovation has third moment $O(\!\sqrt h)$ and
		uniformly bounded higher moments. The $2\times2$ factor coupling preserves
		this order for mixed third moments, and the residual stock quadrature is
		symmetric. The equity decomposition and the local martingale correction give
		the stated first two moments for $Y$ and its mixed covariances. Because each
		state increment is a drift term of order $h$ plus $\sqrt h$ times a
		standardized innovation, the raw third moments are $O(h^2)$, the fourth moments
		are $O(h^2)$, and the fifth absolute moment is $O(h^{5/2})$.
		
	\end{proof}
	
	\begin{proof}[Proof of Theorem~\ref{thm:binomial-weak-order}]
		For a smooth test function, a fourth-order multivariate Taylor expansion and
		Lemma~\ref{lem:binomial-scaled-moments} give
		\eqref{eq:binomial-local-expansion}. The exact discounted semigroup has the
		same expansion, so the one-step defect
		is $O(h^2)$. Insert the defect into the telescoping identity for the numerical
		and exact operator products. Stability transports each local error with a
		uniform factor, and summing $T/h$ defects yields $O(h)$. A fixed finite number
		of regular event interfaces changes only the stability constant.
		
	\end{proof}
	
	\subsection{Delta consistency}
	
	\begin{proof}[Proof of Theorem~\ref{thm:unpruned-binomial-delta}]
		On a regular tube bounded away from $F=0$, the coordinate change
		$Y=\log F$ is smooth, while the fund branch map itself is affine and has the
		branch multiplier as derivative. Polynomial moment bounds therefore make the
		exact and numerical one-step operators stable in the $C^1$ norm. Differentiating the
		fourth-order Taylor expansion used in
		Theorem~\ref{thm:binomial-weak-order} gives a one-step $C^1$ defect of order
		$h^2$. Apply the derivative to the telescoping identity for the numerical and
		exact operator products. Each differentiated defect is transported with a
		uniform $C^1$ stability factor, and summing $T/h$ defects yields $C_Kh$.
		Because the active branch of each contractual map is fixed on the tube, the
		finite set of event interfaces preserves differentiability and changes only the
		constant.
		
	\end{proof}
	
	\subsection{Richardson extrapolation}
	
	\begin{proof}[Proof of Theorem~\ref{thm:richardson-expansion}]
		First consider an interval with no contractual event and write
		$Q_n=Q_{t_n,h}$ and $P_n=P_{t_n,t_{n+1}}$. The telescoping identity is
		\begin{equation}
			Q_0\cdots Q_{N-1}-P_0\cdots P_{N-1}
			=\sum_{n=0}^{N-1}
			Q_0\cdots Q_{n-1}(Q_n-P_n)P_{n+1}\cdots P_{N-1}.
			\label{eq:talay-tubaro-telescope}
		\end{equation}
		Apply it to the terminal value and insert
		Lemma~\ref{lem:richardson-defect}. Stability and the polynomial moment bounds
		allow the local remainder to be summed. The $h^2\Psi_{t_n}$ terms produce
		\begin{equation}
			h^2\sum_{n=0}^{N-1}
			\mathbb E\!\left[
			\Psi_{t_n}u(t_{n+1},Y_{t_{n+1}},V_{t_{n+1}},X_{t_{n+1}})
			\right]
			=hc_1+O(h^2),
		\end{equation}
		because the bracketed quantity is sufficiently regular in time and the sum is
		a Riemann sum. The $h^3\Xi_{t_n}$ terms and the first-order replacement of
		numerical transport by exact semigroups contribute $c_2h^2+o(h^2)$; the
		accumulated $O(h^4)$ local remainder is $O(h^3)$. A contractual date is a
		fixed interface between two such semigroup products. Assumption (R5) preserves
		the regularity needed to repeat the argument over the finitely many annual
		intervals and changes only the coefficients $c_1$ and $c_2$.
		
	\end{proof}
	
\end{document}